%% file: main.tex
\documentclass[11pt]{article}

\usepackage[utf8]{inputenc}
\usepackage{microtype}
\usepackage{amsmath}
\usepackage{amssymb}
\usepackage{amsthm}
\usepackage{bm}
\usepackage{xcolor,colortbl}
\usepackage{dsfont}
\usepackage{authblk}
\usepackage{fullpage}
\usepackage{comment}
\usepackage{nicefrac}
\usepackage{tikz}
\usetikzlibrary{positioning}
\usepackage{mathtools}
\usepackage{bbm}
\usepackage{float}
\usepackage{array,booktabs}
\usepackage[
    backend=biber,
    style=alphabetic,
    maxcitenames=10,
    maxbibnames=99,
    natbib=true,
    url=false,
    doi=false, 
    isbn=false]{biblatex}
\usepackage{hyperref}
\usepackage{mleftright}
\usepackage{thm-restate}
\usepackage[capitalize,noabbrev]{cleveref}
\usepackage[shortlabels]{enumitem}

\newcounter{qst}
\crefname{qst}{Question}{Questions}

\definecolor{Gray}{gray}{0.95}

\usepackage[ruled,vlined,linesnumbered,noend]{algorithm2e}
\makeatletter
\patchcmd\algocf@Vline{\vrule}{\vrule \kern-0.4pt}{}{}
\patchcmd\algocf@Vsline{\vrule}{\vrule \kern-0.4pt}{}{}
\makeatother

\SetKwComment{Hline}{}{\vspace{-3mm}\textcolor{gray}{\hrule}\vspace{1mm}}
\SetKwComment{HlineWd}{}{\vspace{-3mm}\textcolor{gray}{\hrule width 7.5cm}\vspace{1mm}}
\definecolor{darkgrey}{gray}{0.3}
\definecolor{commentcolor}{gray}{0.5}
\SetKwComment{Comment}{\color{commentcolor}[$\triangleright$\ }{}
\SetCommentSty{}
\SetNlSty{}{\color{darkgrey}}{}
\SetKwProg{Fn}{function}{}{}
\SetKwProg{Subr}{subroutine}{}{}
\crefalias{AlgoLine}{line}%
\crefname{algocf}{Algorithm}{Algorithms}

\makeatletter
\let\cref@old@stepcounter\stepcounter
\def\stepcounter#1{%
  \cref@old@stepcounter{#1}%
  \cref@constructprefix{#1}{\cref@result}%
  \@ifundefined{cref@#1@alias}%
    {\def\@tempa{#1}}%
    {\def\@tempa{\csname cref@#1@alias\endcsname}}%
  \protected@edef\cref@currentlabel{%
    [\@tempa][\arabic{#1}][\cref@result]%
    \csname p@#1\endcsname\csname the#1\endcsname}}
\makeatother

\newcommand{\declarecolor}[2]{\definecolor{#1}{RGB}{#2}\expandafter\newcommand\csname #1\endcsname[1]{\textcolor{#1}{##1}}}
\declarecolor{White}{255, 255, 255}
\declarecolor{Black}{0, 0, 0}
\declarecolor{Maroon}{128, 0, 0}
\declarecolor{Coral}{255, 127, 80}
\declarecolor{Red}{182, 21, 21}
\declarecolor{LimeGreen}{50, 205, 50}
\declarecolor{DarkGreen}{0, 100, 0}
\declarecolor{Navy}{0, 0, 128}
\hypersetup{
	colorlinks=true,
	pdfpagemode=UseNone,
	citecolor=DarkGreen,
	linkcolor=Navy,
	urlcolor=Navy,
	pdfstartview=FitW
}

\theoremstyle{plain}
\newtheorem{theorem}{Theorem}[section]
\newtheorem{lemma}[theorem]{Lemma}
\newtheorem{corollary}[theorem]{Corollary}
\newtheorem{proposition}[theorem]{Proposition}
\newtheorem{fact}[theorem]{Fact}

\newtheorem{observation}[theorem]{Observation}
\newtheorem{claim}[theorem]{Claim}

\theoremstyle{definition}
\newtheorem{definition}[theorem]{Definition}

\theoremstyle{remark}
\newtheorem{remark}[theorem]{Remark}

\newcommand*{\N}{{\mathbb{N}}}

\let\R\relax
\newcommand*{\R}{{\mathbb{R}}}
\newcommand*{\E}{{\mathbb{E}}}

\newcommand*{\cX}{{\mathcal{X}}}

\newcommand*{\cR}{{\mathcal{R}}}
\newcommand*{\cN}{{\mathcal{N}}}
\newcommand*{\cA}{{\mathcal{A}}}
\newcommand*{\cC}{{\mathcal{C}}}

\newcommand{\defeq}{\coloneqq}

\DeclareMathOperator{\reg}{Reg}
\DeclareMathOperator{\inter}{int}
\DeclareMathOperator{\relint}{relint}

\DeclareMathOperator{\dom}{dom}
\DeclareMathOperator{\diag}{diag}

\newcommand{\ind}{r}
\newcommand{\mul}{\mu}

\newcommand{\regmin}{\mathfrak{R}}
\newcommand{\Rswap}{\regmin_{swap}}

\newcommand{\optcon}{c^*}
\newcommand{\auxR}{\widetilde{\mathcal{R}}}
\newcommand{\auxu}{\widetilde{\vec{u}}}
\newcommand{\simtrun}{\Delta^{\circ}}
\newcommand{\tree}{\mathcal{T}}
\NewDocumentCommand{\treeset}{o}{\mathbb{T}\IfNoValueF{#1}{_{#1}}}

\newcommand{\rvu}{\texttt{RVU}}

\newcommand{\obsut}{\textsc{ObserveUtility}}
\newcommand{\nextstr}{\textsc{NextStrategy}}
\newcommand{\fp}{\textsc{StationaryDistribution}}

\DeclareMathOperator{\bmoftrl}{\mathtt{BM-OFTRL-LogBar}}

\newcommand{\swapreg}{\mathrm{SwapReg}}

\makeatletter
\AtBeginDocument{%
  \@ifpackageloaded{amsmath}{%
    \newcommand*\patchAmsMathEnvironmentForLineno[1]{%
      \expandafter\let\csname old#1\expandafter\endcsname\csname #1\endcsname
      \expandafter\let\csname oldend#1\expandafter\endcsname\csname end#1\endcsname
      \renewenvironment{#1}%
                       {\linenomath\csname old#1\endcsname}%
                       {\csname oldend#1\endcsname\endlinenomath}%
    }%
    \newcommand*\patchBothAmsMathEnvironmentsForLineno[1]{%
      \patchAmsMathEnvironmentForLineno{#1}%
      \patchAmsMathEnvironmentForLineno{#1*}%
    }%
    \patchBothAmsMathEnvironmentsForLineno{equation}%
    \patchBothAmsMathEnvironmentsForLineno{align}%
    \patchBothAmsMathEnvironmentsForLineno{flalign}%
    \patchBothAmsMathEnvironmentsForLineno{alignat}%
    \patchBothAmsMathEnvironmentsForLineno{gather}%
    \patchBothAmsMathEnvironmentsForLineno{multline}%
  }{}
}
\makeatother

\renewcommand{\vec}[1]{\bm{#1}}
\newcommand{\mat}[1]{\mathbf{#1}}

\newcommand{\range}[1]{[\![#1]\!]}

\newcommand{\tilx}{\widetilde{\vec{x}}}

\DeclarePairedDelimiterX{\card}[1]{\lvert}{\rvert}{#1}
\DeclarePairedDelimiterX{\abs}[1]{\lvert}{\rvert}{#1}
\DeclarePairedDelimiterX{\norm}[1]{\lVert}{\rVert}{#1}
\DeclarePairedDelimiterX{\tuple}[1]{\lparen}{\rparen}{#1}
\DeclarePairedDelimiterX{\parens}[1]{\lparen}{\rparen}{#1}
\DeclarePairedDelimiterX{\brackets}[1]{\lbrack}{\rbrack}{#1}
\DeclarePairedDelimiterX{\set}[1]\{\}{#1}
\let\Pr\relax
\DeclarePairedDelimiterXPP{\Pr}[1]{\mathbb{P}}[]{}{#1}
\DeclarePairedDelimiterXPP{\PrX}[2]{\mathbb{P}_{#1}}[]{}{#2}
\DeclarePairedDelimiterXPP{\Ex}[1]{\mathbb{E}}[]{}{#1}
\DeclarePairedDelimiterXPP{\ExX}[2]{\mathbb{E}_{#1}}[]{}{#2}

\usepackage{lineno}

\usetikzlibrary{fit,shapes.misc,arrows.meta}
\makeatletter
\tikzset{
  fitting node/.style={
    inner sep=0pt,
    fill=none,
    draw=none,
    reset transform,
    fit={(\pgf@pathminx,\pgf@pathminy) (\pgf@pathmaxx,\pgf@pathmaxy)}
  },
  reset transform/.code={\pgftransformreset}
}
\tikzset{cross/.style={path picture={
  \draw[black]
(path picture bounding box.south east) -- (path picture bounding box.north west) (path picture bounding box.south west) -- (path picture bounding box.north east);
}}}
\tikzstyle{ox}=[semithick,draw=black,circle,cross,inner sep=1.2mm]
\makeatother

\input{commands}

\bibliography{main}

\title{Uncoupled Learning Dynamics with $O(\log T)$ Swap Regret in Multiplayer Games}

\author[1]{Ioannis Anagnostides}
\author[2]{Gabriele Farina}
\author[3]{Christian Kroer}
\author[4]{Chung-Wei Lee}
\author[5]{Haipeng Luo}
\author[6]{Tuomas Sandholm}

\affil[1,2,6]{Carnegie Mellon University}
\affil[3]{Columbia University}
\affil[4,5]{University of Southern California}
\affil[6]{Strategy Robot, Inc.}
\affil[6]{Optimized Markets, Inc.}
\affil[6]{Strategic Machine, Inc.}

\affil[ ]{\texttt {\{ianagnos,gfarina\}@cs.cmu.edu}, \texttt{christian.kroer@columbia.edu} \\ \texttt {\{leechung,haipengl\}@usc.edu}, and \texttt{sandholm@cs.cmu.edu}}

\date{\today}

\begin{document}

\maketitle

\pagenumbering{gobble}

\begin{abstract}
    In this paper we establish efficient and \emph{uncoupled} learning dynamics so that, when employed by all players in a general-sum multiplayer game, the \emph{swap regret} of each player after $T$ repetitions of the game is bounded by $O(\log T)$, improving over the prior best bounds of $O(\log^4 (T))$. At the same time, we guarantee optimal $O(\sqrt{T})$ swap regret in the adversarial regime as well. To obtain these results, our primary contribution is to show that when all players follow our dynamics with a \emph{time-invariant} learning rate, the \emph{second-order path lengths} of the dynamics up to time $T$ are bounded by $O(\log T)$, a fundamental property which could have further implications beyond near-optimally bounding the (swap) regret. Our proposed learning dynamics combine in a novel way \emph{optimistic} regularized learning with the use of \emph{self-concordant barriers}. Further, our analysis is remarkably simple, bypassing the cumbersome framework of higher-order smoothness recently developed by Daskalakis, Fishelson, and Golowich (NeurIPS'21).
\end{abstract}

\clearpage
\pagenumbering{arabic}

\input{text/introduction}
\input{text/preliminaries}

\input{text/swap}
\input{text/discussion}

\section*{Acknowledgments}

We are thankful to the anonymous NeurIPS reviewers for many helpful comments. Ioannis Anagnostides is grateful to Dimitris Achlioptas for helpful discussions. 
Christian Kroer is supported by the Office of Naval Research Young Investigator Program under grant N00014-22-1-2530.
Haipeng Luo is supported by the National Science Foundation under grant IIS-1943607 (part of this work was done while he was visiting the Simons Institute for the Theory of Computing).
Tuomas Sandholm is supported by the National Science Foundation under grants IIS-1901403 and CCF-1733556.

\printbibliography
\clearpage

\appendix

\input{text/appendix_prel}
\input{text/appendix_swap}
\input{text/appendix_experiments}

\end{document}

%% file: commands.tex
\newcommand{\nc}{\newcommand}
\newcount\Comments
\nc\noah[1]{\ifnum\Comments=1 {\textcolor{purple}{[ng: #1]}}\fi}
\nc\maxfish[1]{\ifnum\Comments=1{\textcolor{blue}{[mf: #1]}}\fi}
\nc\costis[1]{\ifnum\Comments=1{\textcolor{brown}{[cd: #1]}}\fi}
\nc\costiss[1]{\textcolor{red}{#1}}

\nc\io[1]{\ifnum\Comments=1 {\textcolor{purple}{[ioannis: #1]}}\fi}

\nc{\Opthedge}{OMWU\xspace}

\nc{\DMO}{\DeclareMathOperator}
\nc\old[1]{\textcolor{brown}{[old: #1]}}
\nc{\BR}{\mathbb{R}}
\nc{\BC}{\mathbb{C}}
\DMO{\Bin}{Bin}
\nc{\BN}{\mathbb{N}}
\nc{\distrs}[1]{\Delta({#1})}
\nc{\BZ}{\mathbb{Z}}
\nc{\ep}{\epsilon}
\nc{\ra}{\rightarrow}
\nc{\st}{\star}
\nc{\Reg}[2]{\REG_{{#1},{#2}}}
\nc{\til}{\tilde}
\nc{\kld}[2]{\KL({#1};{#2})}
\nc{\chisq}[2]{\chi^2({#1};{#2})}
\DMO{\POLYLOG}{polylog}

\nc{\matx}[1]{\left(\begin{matrix}#1\end{matrix}\right)}
\DMO{\VAR}{Var}
\DMO{\COV}{Cov}
\nc{\Var}[2]{\VAR_{{#1}}\left({#2}\right)}
\nc{\Cov}[3]{\COV_{{#1}}\left({#2},{#3}\right)}
\DMO{\DD}{D}
\nc{\fd}[2]{\DD_{#1}{#2}}
\nc{\fds}[3]{\left(\fd{#1}{#2}\right)\^{#3}}
\nc{\fdc}[2]{\DD^\circ_{#1}{#2}}
\nc{\fdcs}[3]{\left(\fdc{#1}{#2}\right)\^{#3}}
\nc{\shf}[2]{\EEE_{#1}{#2}}
\nc{\shfs}[3]{\left(\shf{#1}{#2}\right)\^{#3}}
\nc{\normst}[2]{\left\| {#2} \right\|_{#1}^\st}
\renewcommand{\^}[1]{^{(#1)}}
\DeclareMathOperator*{\argmax}{arg\,max}
\DeclareMathOperator*{\argmin}{arg\,min}

\nc{\grad}{\nabla}
\nc{\lng}{\langle}
\nc{\rng}{\rangle}
\nc{\bbone}{\mathbf{1}}
\nc{\bbzero}{\mathbf{0}}
\nc{\MD}{\mathcal{D}}
\nc{\MM}{\mathcal{M}}
\nc{\MZ}{\mathcal{Z}}
\nc{\MU}{\mathcal{U}}
\nc{\MC}{\mathcal{C}}
\nc{\MT}{\mathbb{T}^{n}}
\nc{\MS}{\mathcal{S}}
\nc{\MX}{\mathcal{X}}
\nc{\MY}{\mathcal{Y}}
\nc{\MB}{\mathcal{B}}
\nc{\MJ}{\mathcal{J}}
\nc{\MF}{\mathcal{F}}
\nc{\MG}{\mathcal{G}}
\nc{\MR}{\mathcal{R}}
\nc{\ML}{\mathcal{L}}
\nc{\MQ}{\mathcal{Q}}

\nc{\ba}{\mathbf{A}}
\nc{\bx}{\mathbf{x}}
\nc{\by}{\mathbf{y}}
\nc{\bz}{\mathbf{z}}
\nc{\bs}{\mathbf{s}}
\nc{\bt}{\mathbf{t}}
\nc{\br}{\mathbf{r}}

\nc{\ME}{\mathcal{E}}
\DMO{\View}{View}
\DMO{\KL}{KL}
\nc{\MW}{\mathcal{W}}
\nc{\CS}{\mathscr{S}}
\nc{\CI}{\mathscr{I}}
\nc{\CQ}{\mathscr{Q}}
\nc{\CL}{\mathscr{L}}
\nc{\CM}{\mathscr{M}}
\nc{\CG}{\mathscr{G}}
\nc{\CR}{\mathscr{R}}

\nc{\wh}{\widehat}

\nc{\BM}{BM\xspace}
\nc{\ALG}{\texttt{ALG}}
\nc{\MCT}{{\rm MCT}}

\nc{\matrixLL}{\mat{L}}
\nc{\vectorecks}{\vec{x}}
\nc{\vectorLL}{\vec{\ell}}
\nc{\matrixKYU}{\mat{Q}}
\nc{\rowdot}{\cdot}

%% file: text/introduction.tex
\section{Introduction}

Online learning and game theory share an intricately connected history tracing back to the inception of the modern \emph{no-regret framework} with Robinson's analysis of \emph{fictitious play}~\citep{Robinson51:iterative} and Blackwell's \emph{approachability theorem}~\citep{Blackwell56:analog}. Indeed, the no-regret framework addresses the fundamental question of how independent and decentralized agents can ``learn'' with only limited feedback from their environment, and has led to celebrated connections with game-theoretic equilibrium concepts~\citep{Hart00:Simple,Foster97:Calibrated}. One of the remarkable features of these results is that the learning dynamics are fully \emph{uncoupled}~\citep{Hart00:Simple}: each player is completely agnostic to the utilities of the other players. Thus, there is no communication between the players or any centralized authority dictating behavior throughout the game. Instead, the only ``coordination device'' is the common history of play. An additional desideratum, which is fundamentally tied to the no-regret framework, is what~\citet{Daskalakis11:Near} refer to as \emph{strong uncoupledness}:\footnote{\citet{Daskalakis11:Near} also impose that players are only allowed to (privately) store only a constant number of observed utilities, an assumption also espoused in our work.} players have no information whatsoever about the game (even their own utilities), and they only make decisions based on the utilities received as feedback throughout the repeated game.

In this context, it is well-known that there are broad families of no-regret learning algorithms that, after $T$ repetitions, guarantee regret bounded by $O(\sqrt{T})$, and this bound is known to be insuperable in adversarial environments~\citep{Cesa-Bianchi06:Prediction}. However, this begs the question: \emph{What if the player is not facing adversarial utilities, but instead is competing with other learning agents in a repeated game?} This question was first formulated and addressed by~\citet{Daskalakis11:Near}, who devised strongly uncoupled dynamics converging with a near-optimal rate of $O(\frac{\log T}{T})$ in zero-sum games, a substantial improvement over the $O(1/\sqrt{T})$ rate obtained via traditional approaches within the no-regret framework. Thereafter, there has been a considerable amount of effort in strengthening their result, leading to extensions along several important lines~\citep{Rakhlin13:Optimization,Syrgkanis15:Fast,Chen20:Hedging,Farina19:Stable,Daskalakis21:Near,Anagnostides21:Near,Wei18:More,Foster16:Learning}. In particular, in a recent breakthrough result, \citet{Daskalakis21:Near} showed that when all players in a general game employ an \emph{optimistic} variant of \emph{multiplicative weights update (MWU)} (henceforth \emph{OMWU}), the \emph{external regret} of each player grows as $O(\log^4 (T))$. That result was also subsequently extended to the substantially more challenging performance measure of \emph{swap regret}~\citep{Anagnostides21:Near}. Perhaps the main drawback of the latter results is the complexity of the analysis, relying on establishing a refined property for the dynamics they refer to as \emph{higher-order smoothness}. Our primary contribution in this paper is to develop a novel and much simpler framework, which furthermore improves the prior $O(\log^4(T))$ regret bounds to $O(\log T)$ in general multiplayer games.

\subsection{Overview of Our Contributions}

Before we state our main result, let us first introduce some basic notation. We assume that each player $i \in \range{n}$ selects at every iteration $t$ of the repeated game a probability distribution (mixed strategy) over the set of available actions $\Vec{x}_i^{(t)} \in \Delta(\cA_i)$ (see \Cref{section:prel} for further details). The following theorem is the primary contribution of our work.\footnote{For simplicity in the exposition, we use the $O(\cdot)$ notation in our introduction to suppress parameters that depend (polynomially) on the natural parameters of the game; precise statements are given in \Cref{section:swap}.}

\begin{theorem}[Precise Statement in \Cref{theorem:log-bounded_trajectories}]
    \label{theorem:main}
    There exist strongly uncoupled no-swap-regret learning dynamics so that when employed by all players with learning rate $\eta = \Theta(1)$, the second-order path lengths of the dynamics up to any time $T \in \N$ are bounded by $O(\log T)$; that is, 
    \begin{equation*}
        \sum_{t=1}^T \sum_{i=1}^n \| \Vec{x}_i^{(t)} - \Vec{x}_i^{(t-1)} \|_1^2 = O(\log T).
    \end{equation*}
\end{theorem}

We are not aware of even an $o(T)$ bound for the second-order path lengths---under a \emph{time-invariant} learning rate---prior to our work, except for very restricted classes of games such as zero-sum games. The dynamics of \Cref{theorem:main} combine: (i) the celebrated no-swap-regret template of~\citet{Blum07:From}; (ii) the \emph{optimistic follow the regularizer leader (OFTRL)} algorithm of~\citet{Syrgkanis15:Fast}; and (iii) using a \emph{self-concordant barrier} as a regularizer. The latter was introduced in online learning in the seminal work of~\citet{Abernethy08:Competing}, where the authors obtained the first near-optimal and efficient online learning algorithm for linear bandit optimization; the way we leverage the \emph{log-barrier} in the setting of no-regret learning in games is novel, and crucially leverages the \emph{local norm} induced by the regularizer. The dynamics of \Cref{theorem:main} are also efficiently implementable (see \Cref{remark:numerical}).


The implication of \Cref{theorem:main} is perhaps surprising in view of the inherent \emph{cycling} aspect of no-regret learning in general games. Indeed, it is by now well-understood that any no-regret dynamics will fail to converge---at least for certain games (\emph{e.g.}, see~\citep{Milionis22:Nash}). Nevertheless, \Cref{theorem:main} implies that players will change their strategies \emph{arbitrarily slowly} as the game progresses. As such, players will observe utilities that exhibit very small variation over time, immediately implying near-optimal swap regret.

\begin{corollary}[Precise Statement in \Cref{corollary:adversarial,corollary:near-opt-swap}]
    \label{corollary:main1}
    There exist strongly uncoupled no-swap-regret learning dynamics so that when employed by all players, the individual swap regret of each player is bounded by $O(\log T)$. At the same time, when faced against adversarial utilities each player guarantees $O(\sqrt{T})$ swap regret.
\end{corollary}

\Cref{corollary:main1} improves over the prior best bounds of $O(\log^4 (T))$~\citep{Daskalakis21:Near,Anagnostides21:Near}; a comparison with prior works regarding the algorithm of~\citet{Blum07:From} is given in~\Cref{tab:prior}. In fact, \Cref{corollary:main1} yields, to our knowledge, the first no-regret guarantee in general games for uncoupled methods when players use a \emph{time-invariant} learning rate, a feature that has been extensively motivated in prior works (see, \emph{e.g.},  the discussion in~\citep{Bailey19:Fast}). \Cref{corollary:main1} also establishes near-optimality in the adversarial regime as well, a crucial desideratum in this line of work. Finally, swap regret is a powerful notion of hindsight rationality, trivially subsuming external regret. In particular, in light of well-established connections (see \Cref{theorem:folklore}), we obtain the best known rate of convergence of $O(\frac{\log T}{T})$ to \emph{correlated equilibria} in general games.
\begin{corollary}
    \label{corollary:main2}
    There exist strongly uncoupled learning dynamics so that, when employed by all players, the average correlated distribution of play after $T$ repetitions of the game is an $O(\frac{\log T}{T} )$-approximate correlated equilibrium.
\end{corollary}

\begin{table}[]
    \centering
    \scalebox{.94}{\begin{tabular}{cccc}
    \bf Reference & \bf Algorithm & \bf Swap Regret in Games & \bf Adversarial Swap Regret
    \\
    \toprule
    \citep{Blum07:From} & \emph{E.g.}, BM-MWU & ----- & $O(\sqrt{m \log (m) T})$ \\
    \citep{Chen20:Hedging} & BM-OMWU & $O (\sqrt{n} (m \log (m))^{3/4}) T^{1/4})$ & $\widetilde{O}(\sqrt{m T})$ \\
    \citep{Anagnostides21:Near} & BM-OMWU & $O(n m^4 \log (m) \log^4 (T) )$ & ----- \\
    \rowcolor{Gray}
    \textbf{This paper} & BM-OFTRL-LogBar & $O(n m^{5/2} \log T)$ (\Cref{corollary:near-opt-swap}) & $O(\sqrt{m \log (m) T})$ (\Cref{corollary:adversarial}) \\
    \bottomrule
    \end{tabular}}
    \caption{Prior results regarding the no-swap-regret algorithm of~\citet{Blum07:From} (BM). The second column indicates the algorithm internally employed by the ``master'' BM algorithm; our construction uses OFTRL with log-barrier regularization (\Cref{section:optimistic}). Further, $m$ is the maximum number of actions available to each player. We point out that in the adversarial swap regret bound we have suppressed lower order factors in terms of $T$. We further remark that the near-optimal \emph{internal regret} guarantee of~\citet{Anagnostides21:Near} in turn implies $O(n m \log(m) \log^4(T))$ swap regret for each individual player, but is obtained via the algorithm of~\citet{Stoltz05:Internal}.}
    \label{tab:prior}
\end{table}

From a technical standpoint, our approach is conceptually remarkably simple. Specifically, \Cref{theorem:main} is shown by first establishing the $\rvu$ bound---a fundamental property first identified in~\citep[Definition 3]{Syrgkanis15:Fast}---for swap regret in \Cref{theorem:rvu-swap}; the key ingredient is \Cref{lemma:gamma-term}, which crucially leverages the local norm induced by the log-barrier regularizer over the simplex. Next, \Cref{theorem:main} follows directly by making a seemingly trivial observation: \emph{swap regret is always nonnegative}. A related approach was recently employed in~\citep{Anagnostides22:Last} for external regret, but only works for very restricted classes of games such as zero-sum. As such, we bypasses the cumbersome framework of higher-order smoothness introduced by~\citet{Daskalakis21:Near}.

\subsection{Further Related Work}

The first accelerated dynamics in general games were established by~\citet{Syrgkanis15:Fast}. In particular, they identified a broad class of no-regret learning dynamics---satisfying the so-called $\rvu$ property---for which the \emph{sum} of the players' regrets is $O(1)$. On the other hand, they only obtained an $O(T^{1/4})$ bound for the individual external regret of each player. This is crucial given that the rate of convergence to \emph{coarse} correlated equilibria is driven by the \emph{maximum} of the external regrets. It is important to note that a bound for the sum of the external regrets does not necessarily translate to a bound for the maximum since \emph{external regrets} can be negative. This is in stark contrast to swap regret (\Cref{obseravtion:nonnegative}), a property crucially leveraged in our work. Furthermore, the $O(T^{1/4})$ bounds for the individual external regret in~\citep{Syrgkanis15:Fast} were only recently extended to swap regret by~\citet{Chen20:Hedging}. The main challenge with swap regret---which is also the main focus of our paper---is that the underlying dynamics are much more complex, involving a fixed point operation---namely, the stationary distribution of a Markov chain. Finally, a very intriguing approach for obtaining near-optimal no-external-regret dynamics was recently introduced by~\citet{Piliouras21:Optimal}. The main caveat of that result is that the dynamics they propose are \emph{not} uncoupled, which has been a central desideratum in the line of work on no-regret learning in games. For this reason, the result in~\citep{Piliouras21:Optimal} is not directly comparable with the previous approaches.

%% file: text/preliminaries.tex
\section{Preliminaries}
\label{section:prel}

In this section we introduce the basic background on online optimization and learning in games. For a comprehensive treatment on the subject we refer the interested reader to the excellent book of~\citet{Cesa-Bianchi06:Prediction}.

\paragraph{Conventions} We denote by $\N = \{1, 2, \dots \}$ the set of natural numbers. We use the shorthand notation $\range{n} \defeq \{1, 2, \dots, n\}$. Subscripts are typically used to indicate the player, or a parameter uniquely associated with a player (such as an action available to the player). On the other hand, superscripts are reserved almost exclusively for the (discrete) time index, which is represented via the variable $t$. Also, the $r$-th coordinate of a $d$-dimensional vector $\Vec{x} \in \R^d$ is denoted by $\Vec{x}[r]$. Finally, we let $\log(\cdot)$ be the natural logarithm. 

\subsection{Online Learning and Phi-Regret} 
\label{subsection:online_learning}

Let $\cX \subseteq \R^d$ be a nonempty convex and compact set of strategies, for some $d \in \N$. In the online learning framework the \emph{learner} has to select at every iteration $t \in \N$ a \emph{strategy} $\vec{x}^{(t)} \in \cX$. Then, the environment---be it the ``nature'' or some ``adversary''---returns a (linear) utility function $u^{(t)} :  \cX \ni \vec{x} \mapsto \langle \vec{x}, \vec{u}^{(t)} \rangle$, for some utility vector $\vec{u}^{(t)} \in \R^{d}$, so that the learner receives a utility of $\langle \vec{x}^{(t)}, \vec{u}^{(t)} \rangle$ at time $t$. In the \emph{full information} model the learner receives as feedback the entire utility function, represented by $\vec{u}^{(t)}$. The canonical measure of performance in online learning is based on the notion of \emph{regret}, or more generally, on \emph{Phi-regret}~\citep{Greenwald03:General,Stoltz07:Learning,Gordon08:No}. Formally, for a set of transformations $\Phi: \cX \to \cX$, the \emph{$\Phi$-regret} of a regret minimization algorithm $\regmin$ up to a \emph{time horizon} $T \in \N$ is defined as
\begin{equation}
    \label{eq:Phi-regret}
    \reg_{\Phi}^T \defeq \max_{\phi^* \in \Phi} \left\{ \sum_{t=1}^T \langle \phi^*(\vec{x}^{(t)}), \vec{u}^{(t)} \rangle \right\} - \sum_{t=1}^T \langle \vec{x}^{(t)}, \vec{u}^{(t)} \rangle.
\end{equation}

Naturally, a broader collection of transformations leads to a stronger notion of hindsight rationality; canonical instantiations of Phi-regret include:
\begin{itemize}
    \item[(i)] \emph{External regret} (denoted by $\reg$): $\Phi$ includes only \emph{constant transformations};
    \item[(ii)] \emph{Swap regret} (denoted by $\swapreg$): $\Phi$ includes \emph{all} possible linear transformations.
\end{itemize}

As such, swap regret induces the more powerful notion of hindsight rationality. We point out that our main focus in this paper (\Cref{section:swap}) will be for the special case where $\cX$ is the probability simplex. A crucial property of swap regret is that $\swapreg \geq 0$, as formalized below.
\begin{observation}
    \label{obseravtion:nonnegative}
    Fix any time horizon $T \in \N$. For any sequence of utilities $\vec{u}^{(1)}, \dots, \vec{u}^{(T)}$ and any sequence of strategies $\vec{x}^{(1)}, \dots, \vec{x}^{(T)}$ it holds that $\swapreg^{T} \geq 0$.
\end{observation}

In proof, just consider the identity transformation $\Phi \ni \phi : \vec{x} \mapsto \vec{x}$ in \eqref{eq:Phi-regret}. In contrast, this property does not necessarily hold for external regret.

\subsection{No-Regret Learning and Correlated Equilibria}

A fundamental connection ensures that as long as all players employ no-swap-regret learning dynamics (in the sense that $\swapreg^T = o(T)$), the average correlated distribution of play converges to the set of correlated equilibria~\citep{Hart00:Simple,Foster97:Calibrated,Blum07:From}. Before we formalize this connection, let us first introduce some basic background on games.

\paragraph{Finite Games} Let $\range{n} \defeq \{1, 2, \dots, n\}$ be the set of players, with $n \geq 2$. In a (finite) game, represented in \emph{normal form}, each player $i \in \range{n}$ has a finite set of actions $\cA_i$; for notational simplicity, we will let $m_i \defeq |\cA_i| \geq 2$. For a given joint action profile $\vec{a} = (a_1, \dots, a_n) \in \bigtimes_{i=1}^n \cA_i$, the (normalized) utility received by player $i$ is given by some arbitrary function $u_i : \bigtimes_{i=1}^n \cA_i \to [-1, 1]$. Players are allowed to randomize by selecting a (mixed) strategy $\vec{x}_i \in \Delta(\cA_i) \defeq \left\{ \vec{x} \in \R_{\geq 0}^{|\cA_i|} : \sum_{a_i \in \cA_i} \vec{x}[a_i] = 1 \right\}$; that is, a probability distribution over the available actions. For a joint strategy profile $\vec{x} = (\vec{x}_1, \dots, \vec{x}_n)$, player $i$ receives an \emph{expected utility} of $\E_{\vec{a} \sim \vec{x}}[u_i(\vec{a})] = \sum_{\vec{a} \in \cA} u_i(\vec{a}) \prod_{j \in \range{n}} \vec{x}_j[a_j]$. 

In the problem of no-regret learning in games, every player receives as feedback at time $t \in \N$ a utility vector $\vec{u}_i^{(t)} \in \R^{|\cA_i|}$, so that $\vec{u}_i^{(t)}[a_i] \defeq u_i(a_i; \vec{x}^{(t)}_{-i}) \defeq \E_{\vec{a}_{-i} \sim \vec{x}_{-i}}[u_i(a_i, \vec{a}_{-i})]$, for any $a_i \in \cA_i$; here, we used the notation $\vec{a}_{-i}$ to represent the joint action profile excluding $i$'s component, and analogously for the notation $\vec{x}_{-i}$. \emph{No other information is available to the player}. We are now ready to introduce the concept of a \emph{correlated equilibrium} due to~\citet{Aumann74:Subjectivity}.

\begin{definition}[Correlated Equilibrium~\citep{Aumann74:Subjectivity}]
    \label{def:CE}
    A probability distribution $\vec{\mu}$ over $\bigtimes_{i=1}^n \cA_i$ is an \emph{$\epsilon$-approximate correlated equilibrium}, for $\epsilon \geq 0$, if for any player $i \in \range{n}$ and any swap function $\phi_i : \cA_i \to \cA_i$,
    \begin{equation*}
        \E_{\vec{a} \sim \vec{\mu}}[u_i(\vec{a})] \geq \E_{\vec{a} \sim \vec{\mu}}[u_i(\phi_i(a_i), \vec{a}_{-i})] - \epsilon.
    \end{equation*}
\end{definition}

\begin{theorem}[Folklore]
    \label{theorem:folklore}
    Suppose that each player $i \in \range{n}$ employs a no-swap-regret algorithm such that the cumulative swap regret up to time $T \in \N$ is upper bounded by $\swapreg_i^T$. Further, let $\vec{\mu}^{(t)} \defeq \vec{x}_1^{(t)} \otimes \vec{x}_2^{(t)} \otimes \dots \otimes \vec{x}_n^{(t)}$ be the product distribution at time $t \in \range{T}$, and $\Bar{\vec{\mu}} \defeq \frac{1}{T} \sum_{t=1}^T \vec{\mu}^{(t)}$ be the average correlated distribution of play up to time $T$. Then, $\Bar{\vec{\mu}}$ is a $\max_{i=1}^n \{\swapreg_i^T/T\}$-approximate correlated equilibrium. 
\end{theorem}

Consequently, a central challenge for correlated equilibria is that the rate of convergence is driven by the \emph{maximum} of the swap regrets; this is in contrast to, for example, the rate of convergence of the (utilitarian) social welfare in \emph{smooth games}, which is driven by the \emph{sum} of the players' external regrets~\citep{Syrgkanis15:Fast,Roughgarden15:Intrinsic}.

\section{Optimistic Learning with Self-Concordant Barriers}
\label{section:optimistic}

\emph{Optimistic follow the regularizer leader (OFTRL)}~\citep{Syrgkanis15:Fast} is a \emph{predictive} variant of the standard FTRL paradigm. Specifically, OFTRL maintains an internal \emph{prediction} vector $\Vec{m}^{(t)} \in \R^d$,
and can be expressed with the following update rule for $t \in \N$.
\begin{equation}
    \label{eq:OFTRL}
    \tag{OFTRL}
    \Vec{x}^{(t)} \defeq \argmax_{\Vec{x} \in \cX} \left\{ \Phi^{(t)}(\vec{x}) \defeq \eta \left\langle \Vec{x}, \Vec{m}^{(t)} + \sum_{\tau = 1}^{t-1} \Vec{u}^{(\tau)} \right\rangle - \cR(\Vec{x}) \right\};
\end{equation}
here, $\eta > 0$ serves as the \emph{learning rate}, and $\cR$ is the \emph{regularizer}. For convenience, we also define $\Vec{x}^{(0)} \defeq \argmin_{\Vec{x} \in \cX} \cR(\Vec{x})$. Unless specified otherwise, \eqref{eq:OFTRL} will be instantiated with $\Vec{m}^{(t)} \defeq \Vec{u}^{(t-1)}$, for $t \in \N$. (For convenience in the analysis, and without any loss, we assume that players initially obtain the utilities corresponding to the other players' strategies at time $t = 0$.) 

In~\citep{Syrgkanis15:Fast} the regularizer $\cR$ was assumed to be $1$-strongly convex with respect to some (static) norm $\|\cdot\|$ on $\R^d$. On the other hand, we are introducing an important twist: $\cR$ will be a \emph{self-concordant barrier function} over $\cX$.\footnote{To keep the exposition reasonably self-contained, we give an overview of self-concordant barriers in \Cref{appendix:prel}.} In this context, we first extend (in \Cref{appendix:rvu}) the so-called $\rvu$ bound established in~\citep{Syrgkanis15:Fast} under self-concordant regularization. More precisely, we assume that $\cX$ has nonempty interior $\inter(\cX)$. Further, for $\vec{u} \in \R^d$ the primal local norm with respect to $\vec{x} \in \inter(\cX)$ is defined as $\|\vec{u}\|_{\vec{x}} \defeq \sqrt{\vec{u}^\top \nabla^2 \cR(\vec{x}) \vec{u}}$, while the dual norm is defined as $\|\vec{u}\|_{*, \vec{x}} \defeq \sqrt{\vec{u}^\top (\nabla^2 \cR(\vec{x}))^{-1} \vec{u}}$, assuming that $\cR$ nondegenerate---in the sense that its Hessian is positive definite. Finally, for the purpose of the analysis, we let $\vec{g}^{(t)}$ denote the \emph{be the leader} sequence (see \eqref{eq:BTL} in \Cref{appendix:rvu}); no attempt was made to optimize universal constants.

\begin{restatable}[$\rvu$ for Self-Concordant Regularizers]{theorem}{rvuself}
    \label{theorem:rvu}
    Suppose that $\cR$ is a nondegenerate self-concordant function for $\inter(\cX)$. Moreover, let $\eta > 0$ be such that $\eta \| \Vec{u}^{(t)} - \Vec{m}^{(t)} \|_{*, \Vec{x}^{(t)}} \leq \frac{1}{2}$ and $\eta \| \Vec{m}^{(t)}\|_{*, \Vec{g}^{(t-1)}} \leq \frac{1}{2}$ for all $t \in \range{T}$. Then, the regret of \eqref{eq:OFTRL} under any sequence of utilities $\Vec{u}^{(1)}, \dots, \Vec{u}^{(T)}$ can be bounded as
\begin{equation*}
    \reg^T(\Vec{x}^*) \leq \frac{\cR(\Vec{x}^*)}{\eta} + 2\eta \sum_{t=1}^T \|\Vec{u}^{(t)} - \Vec{m}^{(t)}\|^2_{*, \Vec{x}^{(t)}} - \frac{1}{4\eta} \sum_{t=1}^T \left( \| \Vec{x}^{(t)} - \Vec{g}^{(t)}\|^2_{\Vec{x}^{(t)}} +  \|\Vec{x}^{(t)} - \Vec{g}^{(t-1)} \|^2_{\Vec{g}^{(t-1)}} \right),
\end{equation*}
for any $\Vec{x}^* \in \inter(\cX)$.
\end{restatable}

Here, we also used the standard notation $\reg^T(\vec{x}^*) \defeq \sum_{t=1}^T \langle \vec{x}^* - \vec{x}^{(t)}, \vec{u}^{(t)} \rangle$. Next, we instantiate \Cref{theorem:rvu} using the \emph{log-barrier} on the (probability) simplex: $\cR(\vec{x}) = - \sum_{r=1}^d \log (\vec{x}[r])$. While the probability simplex has empty interior, there is a simple transformation on the relative interior $\relint(\Delta^d)$ that addresses that issue (see \Cref{appendix:rvu}).
\begin{restatable}[$\rvu$ for Log-Barrier on the Simplex]{corollary}{rvusimplex}
    \label{corollary:rvu-simplex}
Suppose that $\cR$ is the log-barrier on the simplex and $\eta \leq \frac{1}{16}$. Then, the regret of \eqref{eq:OFTRL} under any sequence of utilities $\Vec{u}^{(1)}, \dots, \Vec{u}^{(T)}$ can be bounded as
\begin{equation*}
    \reg^T(\Vec{x}^*) \leq \frac{\cR(\Vec{x}^*)}{\eta} + 2 \eta \sum_{t=1}^T \| \Vec{u}^{(t)} - \Vec{u}^{(t-1)} \|^2_{*, \Vec{x}^{(t)}} - \frac{1}{16 \eta} \sum_{t=1}^T \| \Vec{x}^{(t)} - \Vec{x}^{(t-1)} \|^2_{\Vec{x}^{(t-1)}},
\end{equation*}
for any $\vec{x}^* \in \relint(\Delta^d)$, where $\|\vec{x}^{(t)} - \vec{x}^{(t-1)} \|^2_{\vec{x}^{(t-1)}} \defeq \sum_{r=1}^d \left( \frac{\vec{x}^{(t)}[r] - \vec{x}^{(t-1)}[r]}{\vec{x}^{(t-1)}[r]} \right)^2$.
\end{restatable}

We remark that a similar regret bound for \emph{optimistic mirror descent}~\citep{Rakhlin13:Optimization} under log-barrier regularization was shown by~\citep[Theorem 7]{Wei18:More}.

%% file: text/swap.tex
\section{Main Result}
\label{section:swap}

In this section we sketch the proof of our main result, namely \Cref{theorem:main}, leading to \Cref{corollary:main1,corollary:main2}; detailed proofs are deferred to \Cref{appendix:swap}. In this context, we first employ the general template of~\citet{Blum07:From} for constructing a no-swap-regret minimizer $\Rswap$ over the simplex. We proceed with a brief overview of their construction (summarized in \Cref{algo:swap-BM}). In the sequel, we first perform the analysis from the perspective of a single player, without explicitly indicating so in our notation.

\paragraph{The Algorithm of Blum and Mansour} \citet{Blum07:From} construct a ``master'' regret minimization algorithm $\Rswap$ by maintaining a separate and independent external regret minimizer $\regmin_{a}$ for every action $a \in \cA$. To compute the next strategy, $\Rswap$ first obtains the strategy $\vec{x}_a^{(t)} \in \Delta(\cA)$ of $\regmin_a$, for every $a \in \cA$. Then, a (row) stochastic matrix $\mat{Q}^{(t)} \in \mathbb{S}^{|\cA|}$ is constructed, so that the row associated with action $a \in \cA$ is equal to the distribution $\vec{x}^{(t)}_a$, while $\Rswap$ outputs as the next strategy $\vec{x}^{(t)} \in \Delta(\cA)$ any stationary distribution of $\mat{Q}^{(t)}$; that is, $(\mat{Q}^{(t)})^\top \vec{x}^{(t)} = \vec{x}^{(t)}$. Next, upon observing a utility $\Vec{u}^{(t)} \in \R^{|\cA|}$, $\Rswap$ forwards to each individual regret minimizer $\regmin_{a}$ the utility $\vec{u}^{(t)}_a \defeq \Vec{u}^{(t)} \Vec{x}^{(t)}[a] \in \R^{|\cA|}$. \citet{Blum07:From} showed that this algorithm guarantees \emph{no-swap-regret} as long as each individual regret minimizer has sublinear \emph{external regret}; this is formalized in the theorem below.

\begin{theorem}[From External to Swap Regret~\citep{Blum07:From}]
    \label{theorem:swap-external}
Let $\swapreg^T$ be the swap regret of $\Rswap$ and $\reg_a^T$ be the external regret of $\regmin_a$, for each $a \in \cA$, up to time $T \in \N$. Then, 
\begin{equation*}
    \swapreg^T = \sum_{a \in \cA} \reg^T_{a}.
\end{equation*}
\end{theorem}

In this context, we will instantiate each individual regret minimizer $\regmin_a$ with \eqref{eq:OFTRL} under log-barrier regularization---and the same learning rate $\eta > 0$. We will refer to the resulting algorithm as $\bmoftrl$. A central ingredient in our proof of \Cref{theorem:main} is to establish that the resulting no-swap-regret algorithm $\Rswap$ will enjoy an $\rvu$ bound, as stated in \Cref{theorem:rvu-swap}. To this end, we first apply \Cref{corollary:rvu-simplex} for each individual regret minimizer $\regmin_a$, implying that $\swapreg^T = \sum_{a \in \cA} \reg_a^T$ (by \Cref{theorem:swap-external}) is upper bounded by
\begin{equation}
    \label{eq:rvu-Ra}
    \frac{2 m^2 \log T}{\eta} + 2 \eta \sum_{a \in \cA} \sum_{t=1}^T \| \Vec{u}^{(t)} \vec{x}^{(t)}[a] - \Vec{u}^{(t-1)} \vec{x}^{(t-1)}[a] \|^2_{*, \Vec{x}_a^{(t)}} - \frac{1}{16 \eta} \sum_{a \in \cA} \sum_{t=1}^T \| \Vec{x}_a^{(t)} - \Vec{x}_a^{(t-1)} \|^2_{\Vec{x}_a^{(t-1)}}.
\end{equation}

The $\log T$ factor derives from the diameter of the log-barrier regularizer (see \Cref{theorem:diameter}), and appears to be unavoidable using our approach. Now the crux in establishing an $\rvu$ bound for $\Rswap$ is to upper bound the last term in \eqref{eq:rvu-Ra} in terms of the ``movement'' of the stationary distribution. This is exactly where the local norm induced by the log-barrier turns out to be crucial, leading to the following key technical ingredient.
\begin{restatable}{lemma}{gammaterm}
    \label{lemma:gamma-term}
Suppose that each regret minimizer $\regmin_a$ employs \eqref{eq:OFTRL} with log-barrier regularization and $\eta \leq \frac{1}{16}$. Then, for any $t \in \N$,
\begin{equation*}
    \| \vec{x}^{(t)} - \vec{x}^{(t-1)} \|^2_1 \leq 64 |\cA| \sum_{a \in \cA} \| \vec{x}_a^{(t)} - \vec{x}_a^{(t-1)} \|^2_{\vec{x}_a^{(t-1)}}.  
\end{equation*}
\end{restatable}
Intuitively, this lemma ensures that the ``movement'' of the stationary distribution is smooth in terms of the ``movement'' of each row of the transition matrix $\mat{Q}^{(t)}$. To show this, we use the Markov chain tree theorem (\Cref{theorem:MCTT}), which provides a closed-form combinatorial formula for the stationary distribution of an ergodic Markov chain, along with the fact that the log-barrier regularizer guarantees ``multiplicative stability'' of the iterates (\Cref{corollary:mul-stable}). While similar in spirit results have been documented in the literature for dynamics akin to MWU~\citep{Candogan13:Dynamics,Chen20:Hedging}, our proof of \Cref{lemma:gamma-term} crucially hinges on the local norm induced by the log-barrier regularizer. Thus, we are now ready to derive an $\rvu$ bound for swap regret.

\begin{restatable}[$\rvu$ Bound for Swap Regret]{theorem}{rvuswap}
    \label{theorem:rvu-swap}
    Suppose that each $\regmin_a$ employs \eqref{eq:OFTRL} with log-barrier regularization and $\eta \leq \frac{1}{128 \sqrt{m}}$. Then, for $T \geq 2$, the swap regret of $\Rswap$ is
    bounded as
     \begin{equation*}
         \swapreg^T \leq \frac{2 m^2 \log T}{\eta} + 4\eta \sum_{t=1}^T \| \vec{u}^{(t)} - \vec{u}^{(t-1)}\|^2_\infty - \frac{1}{2048 m \eta} \sum_{t=1}^T  \| \Vec{x}^{(t)} - \Vec{x}^{(t-1)}\|_1^2.
     \end{equation*}
\end{restatable}

This theorem follows directly from \eqref{eq:rvu-Ra} and \Cref{lemma:gamma-term}. So far we have focused on bounding the swap regret of each player when faced against arbitrary utilities. Next, we use \Cref{theorem:rvu-swap} to establish a new fundamental property when all players employ the dynamics. Our proof crucially relies on the seemingly insignificant fact that $\swapreg_i^T \geq 0$ (recall \Cref{obseravtion:nonnegative}).

\begin{restatable}[Log-Bounded Second-Order Path Lengths]{theorem}{logbounded}
    \label{theorem:log-bounded_trajectories}
    Suppose that each player $i \in \range{n}$ employs $\bmoftrl$ with $\eta = \frac{1}{128(n-1) \max_{j \in \range{n}} \{\sqrt{m_j}\} }$. Then, for $T \geq 2$,
        \begin{equation*}
    \sum_{i=1}^n \sum_{t=1}^T \| \vec{x}_i^{(t)} - \vec{x}_i^{(t-1)}\|_1^2 \leq 8192 \max_{i \in \range{n}} \{ m_i \} \sum_{i=1}^n m_i^{2} \log T.
\end{equation*}
\end{restatable}

\begin{proof}
Consider any player $i \in \range{n}$. Given that $|u_i(\vec{a})| \leq 1$, for any $\vec{a} \in \cA$ (by the normalization assumption), we have that for any $t \in \range{T}$,
\begin{equation*}
    \|\vec{u}_i^{(t)} - \vec{u}_i^{(t-1)} \|_{\infty} \leq \sum_{\vec{a}_{-i} \in \cA_{-i}} \left| \prod_{j \neq i} \vec{x}_j^{(t)}[a_j] - \prod_{j \neq i} \vec{x}_j^{(t-1)}[a_j] \right| \leq \sum_{j \neq i} \| \vec{x}_j^{(t)} - \vec{x}_j^{(t-1)} \|_1,
\end{equation*}
where we used that the total variation distance between two product distributions is bounded by the sum of the total variations of each individual marginal distribution~\citep{Hoeffding58:Distinguishability}. Thus,
\begin{align*}
    \left( \| \vec{u}_i^{(t)} - \vec{u}_i^{(t-1)} \|_\infty \right)^2 \leq \left( \sum_{j \neq i} \| \vec{x}_{j}^{(t)} - \vec{x}_{j}^{(t-1)} \|_1 \right)^2 \leq (n-1) \sum_{j \neq i} \| \vec{x}_{j}^{(t)} - \vec{x}_{j}^{(t-1)} \|^2_1.
\end{align*}
As a result, using \Cref{theorem:rvu-swap} we conclude that $\sum_{i=1}^n \swapreg_i^T$ can be upper bounded by
\begin{align*}
&2 \log T \sum_{i=1}^n \frac{m_i^2}{\eta} + 4 \eta (n-1) \sum_{i =1}^n \sum_{j \neq i} \sum_{t=1}^T \| \vec{x}_{j}^{(t)} - \vec{x}_{j}^{(t-1)} \|^2_1 - \sum_{i = 1}^n \frac{1}{2048 m_i \eta} \sum_{t=1}^T  \| \Vec{x}_i^{(t)} - \Vec{x}_i^{(t-1)}\|_1^2 \\
=\, &2 \log T \sum_{i=1}^n \frac{m_i^2}{\eta} + \sum_{i=1}^n \left( 4\eta (n-1)^2 - \frac{1}{2048 m_i \eta} \right) \sum_{t=1}^T \| \vec{x}_i^{(t)} - \vec{x}_i^{(t-1)} \|_1^2 \\
\leq\, &2 \log T \sum_{i=1}^n \frac{m_i^2}{\eta} - \frac{1}{4096} \sum_{i=1}^n \frac{1}{m_i \eta} \sum_{t=1}^T \| \vec{x}_i^{(t)} - \vec{x}_i^{(t-1)} \|_1^2,
\end{align*}
since $\eta = \frac{1}{128 (n-1) \max_{j \in \range{n}}\{ \sqrt{m_j}\}}$. But, given that $0 \leq \sum_{i=1}^n \swapreg_i^T$, we conclude that
\begin{equation*}
    \frac{1}{\max_{i \in \range{n}} \{ \sqrt{m_i}\}} \sum_{i=1}^n \sum_{t=1}^T \| \vec{x}_i^{(t)} - \vec{x}_i^{(t-1)}\|_1^2 \le 8192 \max_{i \in \range{n}} \{ \sqrt{m_i}\} \sum_{i=1}^n m_i^{2} \log T.
\end{equation*}
\end{proof}

We are not aware of even $o(T)$ bounds for the second-order path lengths in prior works (using a time-invariant learning rate), except in very restricted classes of games such as zero-sum and potential games~\citep{Anagnostides22:Last}. An example of the implication of \Cref{theorem:log-bounded_trajectories} in a variant of \emph{Shapley's game}~\citep{Shapley64:Some,Daskalakis10:On} is illustrated in \Cref{fig:Shapley}. Although the dynamics appear to cycle, and the Nash gap---the maximum of the best response gaps---is always large, the players are changing their (mixed) strategies with gradually diminishing speed; further discussion and experiments are included in \Cref{appendix:experiments}. 

\begin{figure}[!ht]
    \centering
    \includegraphics[scale=.65]{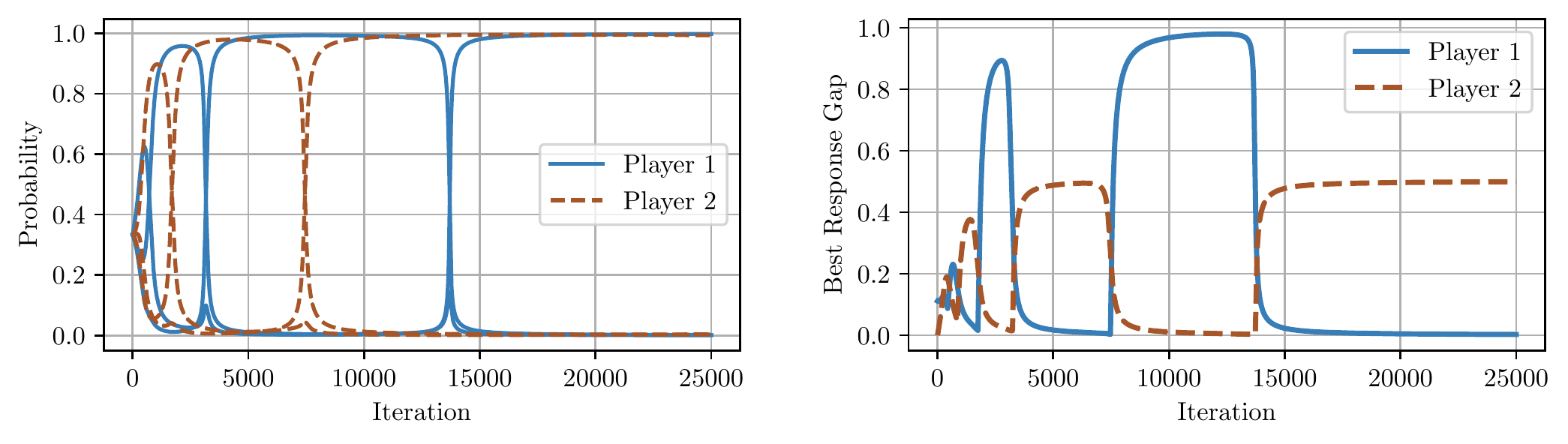}
    \caption{The trajectories of the $\bmoftrl$ algorithm.}
    \label{fig:Shapley}
\end{figure}

As an immediate consequence, combining \Cref{theorem:log-bounded_trajectories} with \Cref{theorem:rvu-swap} implies near-optimal individual swap regret.

\begin{restatable}[Near-Optimal Individual Swap Regret]{corollary}{optswap}
    \label{corollary:near-opt-swap}
    Suppose that all players use $\bmoftrl$ with $\eta = \frac{1}{128 (n-1) \max_{j \in \range{n}} \{\sqrt{m_j}\}}$. Then, the individual swap regret $\swapreg_i^T$ up to time $T \geq 2$ of each player $i \in \range{n}$ can be bounded as
    \begin{equation*}
        \swapreg_i^T \leq 256 \max_{j \in \range{n}} \{ \sqrt{m_j} \} \left( (n-1) m_i^{2} + \sum_{j=1}^n m_j^{2} \right) \log T.
    \end{equation*}
\end{restatable}

We point out that our distributed protocol makes the very mild assumption that each player knows an upper bound on the total number of players and the maximum number of actions in order to appropriately tune the learning rate. Further, as is the case with the result in~\citep{Daskalakis21:Near}, the individual regret of each player predicted by \Cref{corollary:near-opt-swap} grows linearly with the number of players. This can be unsatisfactory in games with a large number of players---\emph{i.e.}, $n \gg 1$. For this reason, in \Cref{theorem:largegames} we refine our guarantee in games where the utility of each player depends only a small number of other players.

Finally, we adapt the learning dynamics so that each player guarantees at the same time near-optimal swap regret in the adversarial regime as well.

\begin{restatable}[Adversarial Robustness]{corollary}{adverob}
    \label{corollary:adversarial}
    There exist dynamics such that when all players follow them the individual swap regret of each player grows as in \Cref{corollary:near-opt-swap}. Moreover, when faced against adversarial utilities, such that $\|\vec{u}_i^{(t)}\|_{\infty} \leq 1$ for all $t \in \range{T}$, the algorithm guarantees that \begin{equation*}
     \swapreg_i^T \le 256 \max_{j \in \range{n}} \{ \sqrt{m_j} \} \left( (n-1) m_i^{2} + \sum_{j=1}^n m_j^{2} \right) \log T + 2 \sqrt{m_i \log m_i T} + 2.
    \end{equation*}
\end{restatable}
Our adaptation is particularly natural: If all players follow the prescribed protocol, \Cref{theorem:log-bounded_trajectories} implies that the observed utilities of each player $i$ will be such that $\sum_{\tau=1}^t \| \vec{u}_i^{(\tau)} - \vec{u}_i^{(\tau-1)} \|_{\infty} = O(\log t)$. So, if at any time the player identifies that the previous condition was violated, it suffices to switch to a no-swap-regret minimizer (such as BM-MWU) tuned to face advarsarial losses---in which case it is crucial to use a vanishing learning rate $\eta = O(1/\sqrt{T})$. 

\begin{remark}[Numerical Precision]
    \label{remark:numerical}
    As is standard, we assumed that the iterates of \eqref{eq:OFTRL} were computed exactly, without taking into account issues relating to numerical precision. To justify this, one can use \emph{Damped Newton's method} in order to determine an $\epsilon$-nearby point to the optimal in $O(\log \log (1/\epsilon))$ iterations~\citep{Nemirovski08:Interior}. This would extend all the regret bounds with up to an $O(\epsilon T)$ error. So, with only $O(\log \log T)$ repetitions of Damped Newton's method (per iteration) the error in the regret bounds becomes $O(1)$, and all of our guarantees immediately extend; see~\citep[Appendix A.5]{Farina22:Near} for an analogous extension under approximate iterates.
\end{remark}

%% file: text/discussion.tex
\section{Discussion}
\label{section:discussion}

Our main contribution in this paper was to establish a fundamental new property characterizing the trajectories of certain uncoupled no-regret learning dynamics, summarized in \Cref{theorem:main}. This property directly guarantees the best known and near-optimal bound of $O(\log T)$ for the swap regret incurred by each player in a general multiplayer game. Investigating further consequences of \Cref{theorem:main} is an interesting direction for the future. 
We also believe that our framework could have new implications for learning in games with partial information; \emph{e.g.}, see \citep{Wei18:More}. Another interesting avenue is to extend our scope to more general and combinatorial sets beyond the probability simplex, in order to (efficiently) encompass, for example, games in \emph{extensive form}.

Further, our no-swap-regret learning dynamics have external regret trivially bounded according to \cref{corollary:near-opt-swap}. Consequently, our construction yields no-external-regret learning dynamics with a more favorable dependence on $T$ compared to \citep{Daskalakis21:Near} ($\log T$ compared to the $\log^4 (T)$ of the latter), but with a worse dependence on the number of actions (polynomial rather than logarithmic). Our method also has higher per-iteration complexity. For these reasons, extending the scope of our framework beyond self-concordant regularization is an important direction for future research. Indeed, we conjecture that OMWU has \emph{bounded} second-order path lengths, a property that would imply the first uncoupled learning dynamics with bounded regret, but establishing that likely requires new insights.

%% file: text/appendix_prel.tex
\section{Preliminaries on Self-Concordant Barriers}
\label{appendix:prel}

In this section we provide the necessary background on self-concordant barriers. For a more comprehensive overview on the theory of self-concordant barriers and their role in interior-point methods we refer to the book of~\citet{Nesterov04:Introductory}, the lecture notes of~\citet{Nemirovski04:Interior}, as well as the survey of~\citet{Nemirovski08:Interior}. We start this section by introducing the central concept of a \emph{self-concordant function}.

\subsection{Self-Concordant Functions}

\begin{definition}[Self-Concordant Function]
    \label{definition:self-concordant}
Let $Q \subseteq \R^d$ be a nonempty open and convex set. A convex function $f : Q \to \R$ in $\cC^3$ is called \emph{self-concordant} on $Q$ if it satisfies the following properties.
\begin{itemize}
    \item[(i)] (Barrier property) For every sequence $( \Vec{x}_i \in Q )_{i=1}^{\infty}$ converging to a boundary point of $Q$ as $i \to \infty$ it holds that $f(\Vec{x}_i) \to \infty$; 
    \item[(ii)] (Differential inequality of self-concordance) $f$ satisfies the
    inequality 
    \begin{equation}
        \label{eq:diff-ineq}
    |D^3 f(\Vec{x})[\Vec{u}, \Vec{u}, \Vec{u}]| \leq 2 \left( D^2 f(\Vec{x})[\Vec{u},\Vec{u}] \right)^{3/2},
    \end{equation}
    for all $\Vec{x} \in Q$ and $\Vec{u} \in \R^d$.
\end{itemize}
\end{definition}
In \eqref{eq:diff-ineq} we used the notation
\begin{equation*}
    D^k f(\Vec{x})[\Vec{u}_1, \dots, \Vec{u}_k] \defeq  \left. \frac{\partial^k}{\partial s_1 \dots \partial s_k} \right|_{s_1 = \dots = s_k = 0} f(\Vec{x} + s_1 \Vec{u}_1 + \dots + s_k \Vec{u}_k)
\end{equation*}
to denote the $k$-th-order differential of $f$ at point $\Vec{x}$ along the directions $\Vec{u}_1, \Vec{u}_2, \dots, \Vec{u}_k$. Self-concordance, in the sense of \Cref{definition:self-concordant}, basically imposes a Lipschitz-continuity condition on the Hessian of $f$, but with respect to the \emph{local norm} induced by the Hessian itself~\citep{Nemirovski04:Interior}. One may allow \eqref{eq:diff-ineq} to hold with a multiplicative factor $M_f \geq 0$ on the right hand side, in which case $f$ is said to be self-concordant with parameter $M_f$; unless explicitly specified otherwise, it will be assumed that $M_f = 1$. As a concrete example, we point out that the logarithmic barrier for the nonnegative ray, namely the univariate function $f : (0, +\infty) \ni x \mapsto - \log x$, is self-concordant (with parameter $M_f = 1$). 

A crucial fact is that self-concordance is preserved under any linear perturbation, as can be verified directly from \Cref{definition:self-concordant}. We also point out a certain property which will be useful when composing different functions, and is also an immediate consequence of \Cref{definition:self-concordant}.

\begin{lemma}[\citep{Nemirovski04:Interior}]
\label{lemma:comp}
Let $f_i$ be self-concordant on $\dom f_i$, for all $i \in \range{k}$. Then, assuming that $\dom f \defeq \cap_{i=1}^k \dom f_i \neq \emptyset$, the function $f(\Vec{x}) \defeq \sum_{i=1}^k f_i(\Vec{x})$ is self-concordant.
\end{lemma}


\subsection{Useful Inequalities}

Let $f$ be a self-concordant function. In the sequel we will tacitly assume that $f$ is \emph{nondegenerate}, in the sense that the Hessian $\nabla^2 f(\Vec{x})$ is positive definite, for any $\Vec{x} \in \dom f$. In this context, we define $\|\Vec{u}\|_{f, \Vec{x}} \defeq \sqrt{\Vec{u}^\top \nabla^2 f(\Vec{x}) \Vec{u}}$ to be the \emph{(primal) local norm} of direction $u$ induced by $f$ at point $\Vec{x} \in \dom f$. (It is easy to verify that $\| \Vec{u}\|_{f, \vec{x}}$ indeed satisfies the axioms of a norm.) To lighten our notation, we will oftentimes simply write $\|\Vec{u}\|_{\Vec{x}}$ when the underlying self-concordant function is clear from the context. 
%
The following inequality will be used to derive quadratic growth bounds with respect to the minimum of a self-concordant function.

\begin{lemma}[\citep{Nesterov04:Introductory}]
    \label{lemma:QG}
Let $f$ be a self-concordant function. Then, for any $\Vec{x}, \tilx \in \dom f$,
\begin{equation*}
    f(\tilx) \geq f(\Vec{x}) + \langle \nabla f(\Vec{x}), \tilx - \Vec{x} \rangle + \omega \left( \|\tilx - \Vec{x}\|_{\Vec{x}} \right),
\end{equation*}
where $\omega(s) \defeq s - \log(1 + s)$.
\end{lemma}

It will be convenient to use a quadratic lower bound for $\omega(s)$, as implied by the following simple fact.

\begin{fact}
    \label{fact:omega-lb}
    Let $\omega(s) = s - \log(1 + s)$. Then, 
    \begin{equation*}
        \omega(s) \geq \frac{s^2}{2(1 + s)}.
    \end{equation*}
    In particular, for $s \in [0,1]$ it holds that $\omega(s) \geq \frac{s^2}{4}$.
\end{fact}

Next, let us consider the optimization problem associated with the minimization of a self-concordant function, namely

\begin{equation}
    \label{eq:opt-self}
    \min \{ f(\Vec{x}) : \Vec{x} \in \dom f \},
\end{equation}
for a self-concordant $f$. The \emph{Newton Decrement} of $f$ at point $\Vec{x} \in \dom f$ is defined as 
\begin{equation*}
    \lambda(\Vec{x}, f) \defeq \| \nabla f(\Vec{x}) \|_{*, \Vec{x}} = \sqrt{(\nabla f(\Vec{x}))^\top (\nabla^2 f(\Vec{x}))^{-1} \nabla f(\Vec{x})}.
\end{equation*}

The following result guarantees (existence and) uniqueness for the optimization problem \eqref{eq:opt-self}.

\begin{lemma}[\citep{Nesterov04:Introductory}]
    \label{lemma:uniqueness}
    Let $f$ be a self-concordant function such that $\lambda(\Vec{x}, f) < 1$, for some $\Vec{x} \in \dom f$. Then, the optimization problem \eqref{eq:opt-self} has a unique solution.
\end{lemma}
 
Assuming that $\cX$ is a convex and compact set with nonempty interior, we will also use the following important fact.

\begin{lemma}[\citep{Nemirovski04:Interior}]
    \label{lemma:stability-lam}
Let $f : \inter(\cX) \to \R$ be a self-concordant function with $\Vec{x}^* \defeq \argmin_{\Vec{x}} f(\Vec{x})$, and some $\Vec{x} \in \inter(\cX)$. Then, if $\lambda(\Vec{x}, f) \le \frac{1}{2}$,
\begin{align*}
    \|\Vec{x} - \Vec{x}^*\|_{\Vec{x}} \leq 2 \lambda(\Vec{x}, f); \\
    \|\Vec{x} - \Vec{x}^*\|_{\Vec{x}^*} \leq 2 \lambda(\Vec{x}, f).
\end{align*}
\end{lemma}

\subsection{Self-Concordant Barriers}

Next, we introduce the concept of a \emph{self-concordant barrier}.

\begin{definition}[Self-Concordant Barrier]
    \label{definition:self-concordant-barrier}
    Let $\cX \subseteq \R^d$ be a convex and compact set with nonempty interior $\inter(\cX)$ (domain). A function $f : \inter (\cX) \to \R$ is called a \emph{$\theta$-self-concordant barrier} for $\cX$ if
    \begin{itemize}
        \item[(i)] $f$ is self-concordant on $\inter (\cX)$; and
        \item[(ii)] for all $\Vec{x} \in \inter(\cX)$ and $\Vec{u} \in \R^d$,
        \begin{equation}
            \label{eq:diff-bar}
            | D f(\Vec{x})[\Vec{u}]| \leq \theta^{1/2} \left( D^2 f(\Vec{x})[\Vec{u}, \Vec{u}] \right)^{1/2}.
        \end{equation}
    \end{itemize}
\end{definition}
We note that \eqref{eq:diff-bar} imposes that $f$ is Lipshitz continuous with parameter $\theta^{1/2}$, but with respect to the local Euclidean metric induced by the Hessian. As an example, it is immediate to see that the function $\cR(x) \defeq - \log x$ is a $1$-self-concordant barrier for the nonnegative ray. The following lemma will be useful when composing self-concordant barriers.

\begin{lemma}[\citep{Nesterov04:Introductory}]
    \label{lemma:sum-bar}
    Let $f_i$ be a $\theta_i$-self-concordant barrier for the compact and convex domain $\cX_i \subseteq \R^d$, for all $i \in \range{k}$. If the set $\cX \defeq \cap_{i \in \range{k}} \cX_i$ has nonempty interior, the function $f(\Vec{x}) \defeq \sum_{i=1}^k f_i(\Vec{x})$ is a $\left( \sum_{i=1}^k \theta_i \right)$-self-concordant barrier for $\cX$.
\end{lemma}

\paragraph{Minkowski Function} Finally, we will require the fact that a self-concordant barrier does not grow overly quickly close to the boundary of $\cX$. In particular, the growth is only \emph{logarithmic} as a function of the inverse distance from the boundary. To formalize this, let us introduce the \emph{Minkowski function} on $\cX$, defined as follows.
\begin{equation*}
    \pi(\tilx ; \Vec{x}) = \inf \left\{ s \geq 0 : \Vec{x} + s^{-1} (\tilx - \Vec{x}) \in \cX \right\}.
\end{equation*}

We remark that $\pi(\tilx ; \Vec{x}) \in [0,1]$. When $\Vec{x}$ is the ``center'' of $\cX$, $\pi(\tilx ; \Vec{x})$ can be thought of as the distance of $\tilx$ from the boundary of $\cX$. In this context, we will use the following theorem. 

\begin{theorem}
    \label{theorem:diameter}
    For any $\theta$-self-concordant barrier $\cR$ on $\cX$ and $\Vec{x}, \tilx \in \inter (\cX)$,
    \begin{equation*}
        \cR(\tilx) - \cR(\Vec{x}) \leq \theta \log \left( \frac{1}{1 - \pi(\tilx ; \Vec{x})} \right).
    \end{equation*}
\end{theorem}

\section{RVU Bounds under Self-Concordant Barriers}
\label{appendix:rvu}

In this section we establish the $\rvu$ property~\citep{Syrgkanis15:Fast} for \eqref{eq:OFTRL} when the regularizer is a self-concordant function. The main result of this section is \Cref{theorem:rvu}, while \Cref{corollary:rvu-simplex} is an instantiation on the probability simplex.

As usual, for the purpose of our analysis we consider the auxiliary \emph{be the leader (BTL) sequence}, defined for $t \in \N \cup \{0\}$ as follows.

\begin{equation}
    \label{eq:BTL}
    \tag{BTL}
    \Vec{g}^{(t)} \defeq \argmax_{\Vec{g} \in \cX} \left\{ \Psi^{(t)}(\vec{g}) \defeq \eta \left\langle \Vec{g}, \sum_{\tau=1}^{t} \Vec{u}^{(\tau)} \right\rangle - \cR(\Vec{g}) \right\}.
\end{equation}

By convention, we have let $\Vec{g}^{(0)} \defeq \argmin_{\Vec{g} \in \cX} \cR(\Vec{g})$. We also remark that, as long as $\eta \| \Vec{u}^{(t)} - \Vec{m}^{(t)} \|_{*, \Vec{x}^{(t)}} \leq \frac{1}{2}$ and $\eta \| \Vec{m}^{(t)}\|_{*, \Vec{g}^{(t-1)}} \leq \frac{1}{2}$, for all $t \in \range{T}$, both \eqref{eq:BTL} and \eqref{eq:OFTRL} are well-posed, as can be verified using \Cref{lemma:uniqueness} (see \Cref{lemma:stability}). For convenience, and without any loss of generality, in the sequel it is assumed that $\cR$ is normalized so that $\min_{\vec{x}} \cR(\Vec{x}) = 0$. We are now ready to establish the following theorem.
\begin{theorem}
    \label{theorem:prel-rvu}
Suppose that $\cR$ is a nondegenerate self-concordant function for $\inter(\cX)$, and let $\eta > 0$ be such that $\eta \| \Vec{u}^{(t)} - \Vec{m}^{(t)} \|_{*, \Vec{x}^{(t)}} \leq \frac{1}{2}$ and $\eta \| \Vec{m}^{(t)}\|_{*, \Vec{g}^{(t-1)}} \leq \frac{1}{2}$, for all $t \in \range{T}$. Then, the regret of \eqref{eq:OFTRL} $\reg^T(\vec{x}^*)$ with respect to any $\vec{x}^* \in \inter(\cX)$ and under any sequence of utilities $\Vec{u}^{(1)}, \dots, \Vec{u}^{(T)}$ can be bounded as
\begin{equation*}
    \frac{\cR(\Vec{x}^*)}{\eta} + \sum_{t=1}^T \|\Vec{u}^{(t)} - \Vec{m}^{(t)}\|_{*, \Vec{x}^{(t)}} \|\Vec{x}^{(t)} - \Vec{g}^{(t)} \|_{\Vec{x}^{(t)}} - \frac{1}{\eta} \sum_{t=1}^T \left( \omega( \| \Vec{x}^{(t)} - \Vec{g}^{(t)}\|_{\Vec{x}^{(t)}}) + \omega( \|\Vec{x}^{(t)} - \Vec{g}^{(t-1)} \|_{\Vec{g}^{(t-1)}}) \right),
\end{equation*}
where $\omega(\cdot)$ is defined as in \Cref{lemma:QG}.
\end{theorem}

\begin{proof}
The proof proceeds similarly to \citep[Theorem 19]{Syrgkanis15:Fast}. The first observation is that
\begin{equation*}
    \langle \Vec{x}^* - \Vec{x}^{(t)}, \Vec{u}^{(t)} \rangle = \langle \Vec{g}^{(t)} - \Vec{x}^{(t)}, \Vec{u}^{(t)} - \Vec{m}^{(t)} \rangle + \langle \Vec{g}^{(t)} - \Vec{x}^{(t)}, \Vec{m}^{(t)} \rangle + \langle \Vec{x}^* - \Vec{g}^{(t)}, \Vec{u}^{(t)} \rangle.
\end{equation*}
Given that $\langle \Vec{g}^{(t)} - \Vec{x}^{(t)}, \Vec{u}^{(t)} - \Vec{m}^{(t)} \rangle \leq \|\Vec{u}^{(t)} - \Vec{m}^{(t)}\|_{*, \Vec{x}^{(t)}} \|\Vec{x}^{(t)} - \Vec{g}^{(t)} \|_{\Vec{x}^{(t)}}$, by H\"older's inequality, it suffices to prove that for any $T \in \N$ and $\Vec{x}^* \in \inter(\cX)$,
\begin{align}
    \sum_{t=1}^T \left( \langle \Vec{g}^{(t)} - \Vec{x}^{(t)}, \Vec{m}^{(t)} \rangle + \langle \Vec{x}^* - \Vec{g}^{(t)}, \Vec{u}^{(t)} \rangle \right) &\leq \frac{\cR(\Vec{x}^*)}{\eta} \notag \\ -\frac{1}{\eta} \sum_{t=1}^T& \left( \omega( \| \Vec{x}^{(t)} - \Vec{g}^{(t)}\|_{\Vec{x}^{(t)}}) + \omega( \|\Vec{x}^{(t)} - \Vec{g}^{(t-1)} \|_{\Vec{g}^{(t-1)}}) \right). \label{eq:ind}
\end{align}

We will establish this claim via induction. For convenience, we use as base for the induction the case where $T = 0$, in which case \eqref{eq:ind} holds trivially since $\cR(\Vec{x^*}) \geq 0$ for any $\Vec{x}^* \in \inter(\cX)$.\footnote{By convention, it is assumed that a sum over an empty set is $0$.} Now for the inductive step, assume that for some $T \in \{0, 1, \dots\}$,
\begin{align}
    \sum_{t=1}^T \left( \langle \Vec{g}^{(t)} - \Vec{x}^{(t)}, \Vec{m}^{(t)} \rangle - \langle \Vec{g}^{(t)}, \Vec{u}^{(t)} \rangle \right) &\leq - \sum_{t=1}^T \langle \Vec{x}^*, \Vec{u}^{(t)} \rangle + \frac{\cR(\Vec{x}^*)}{\eta} \notag \\- \frac{1}{\eta} &\sum_{t=1}^T \left( \omega( \| \Vec{x}^{(t)} - \Vec{g}^{(t)}\|_{\Vec{x}^{(t)}}) + \omega( \|\Vec{x}^{(t)} - \Vec{g}^{(t-1)} \|_{\Vec{g}^{(t-1)}}) \right), \label{eq:ind-new}
\end{align}
for any $\Vec{x}^* \in \inter(\cX)$. We will prove the claim for $T + 1$. Indeed, applying \eqref{eq:ind-new} for $\Vec{x}^* = \Vec{g}^{(T)}$ and adding on both sides the term $\langle \Vec{g}^{(T+1)} - \Vec{x}^{(T+1)}, \Vec{m}^{(T+1)} \rangle - \langle \Vec{g}^{(T+1)}, \Vec{u}^{(T+1)} \rangle$ yields that 
\begin{align}
    \sum_{t=1}^{T+1} &\left( \langle \Vec{g}^{(t)} - \Vec{x}^{(t)}, \Vec{m}^{(t)} \rangle - \langle \Vec{g}^{(t)}, \Vec{u}^{(t)} \rangle \right) \leq \notag
    \\ &- \left\langle \Vec{g}^{(T)}, \sum_{t=1}^T \Vec{u}^{(t)} \right\rangle + \frac{\cR(\Vec{g}^{(T)})}{\eta} + \langle \Vec{g}^{(T+1)} - \Vec{x}^{(T+1)}, \Vec{m}^{(T+1)} \rangle - \langle \Vec{g}^{(T+1)}, \Vec{u}^{(T+1)} \rangle \notag \\ &- \frac{1}{\eta} \sum_{t=1}^T \left( \omega( \| \Vec{x}^{(t)} - \Vec{g}^{(t)}\|_{\Vec{x}^{(t)}}) + \omega( \|\Vec{x}^{(t)} - \Vec{g}^{(t-1)} \|_{\Vec{g}^{(t-1)}}) \right). \label{align:ind-first}
\end{align}

Now, by the first-order optimality condition of the optimization problem associated with \eqref{eq:BTL}, we have that $\nabla \Psi^{(T)}(\Vec{g}^{(T)}) = \Vec{0}$. As a result, \Cref{lemma:QG} implies that
\begin{align}
    - \Psi^{(T)}(\Vec{x}^{(T+1)}) + \Psi^{(T)}(\Vec{g}^{(T)}) \geq \omega(\|\Vec{x}^{(T+1)} - \Vec{g}^{(T)}\|_{\Vec{g}^{(T)}}) \iff \notag \\ 
    \!\!- \left\langle \Vec{x}^{(T+1)}, \sum_{t=1}^T \Vec{u}^{(t)} \right\rangle + \frac{\cR(\Vec{x}^{(T+1)})}{\eta} + \left\langle \Vec{g}^{(T)}, \sum_{t=1}^T \Vec{u}^{(t)} \right\rangle - \frac{\cR(\Vec{g}^{(T)})}{\eta} \geq \frac{1}{\eta} \omega(\|\Vec{x}^{(T+1)} - \Vec{g}^{(T)}\|_{\Vec{g}^{(T)}}), \label{align:QG-psi}
\end{align}
where we used the fact that $- \Psi^{(T)}$ is a self-concordant function, which in turn follows directly from the fact that linear perturbations do not affect self-concordance. Thus, plugging \eqref{align:QG-psi} to \eqref{align:ind-first} yields that 
\begin{align}
    \sum_{t=1}^{T+1} &\left( \langle \Vec{g}^{(t)} - \Vec{x}^{(t)}, \Vec{m}^{(t)} \rangle - \langle \Vec{g}^{(t)}, \Vec{u}^{(t)} \rangle \right) \leq \notag
    \\ &- \left\langle \Vec{x}^{(T+1)}, \sum_{t=1}^T \Vec{u}^{(t)} \right\rangle + \frac{\cR(\Vec{x}^{(T+1)})}{\eta} + \langle \Vec{g}^{(T+1)} - \Vec{x}^{(T+1)}, \Vec{m}^{(T+1)} \rangle - \langle \Vec{g}^{(T+1)}, \Vec{u}^{(T+1)} \rangle \notag \\ &- \frac{1}{\eta} \sum_{t=1}^T \left( \omega( \| \Vec{x}^{(t)} - \Vec{g}^{(t)}\|_{\Vec{x}^{(t)}}) + \omega( \|\Vec{x}^{(t)} - \Vec{g}^{(t-1)} \|_{\Vec{g}^{(t-1)}}) \right) - \frac{1}{\eta} \omega(\|\Vec{x}^{(T+1)} - \Vec{g}^{(T)}\|_{\Vec{g}^{(T)}}) \notag \\
    &= - \left\langle \Vec{x}^{(T+1)}, \Vec{m}^{(T+1)} + \sum_{t=1}^T \Vec{u}^{(t)} \right\rangle + \frac{\cR(\Vec{x}^{(T+1)})}{\eta} + \langle \Vec{g}^{(T+1)}, \Vec{m}^{(T+1)} \rangle - \langle \Vec{g}^{(T+1)}, \Vec{u}^{(T+1)} \rangle \notag \\ &- \frac{1}{\eta} \sum_{t=1}^T \left( \omega( \| \Vec{x}^{(t)} - \Vec{g}^{(t)}\|_{\Vec{x}^{(t)}}) + \omega( \|\Vec{x}^{(t)} - \Vec{g}^{(t-1)} \|_{\Vec{g}^{(t-1)}}) \right) - \frac{1}{\eta} \omega(\|\Vec{x}^{(T+1)} - \Vec{g}^{(T)}\|_{\Vec{g}^{(T)}}). \label{align:ind-second}
\end{align}
Similarly, by the first-order optimality condition of the optimization problem associated with \eqref{eq:OFTRL}, we have that $\nabla \Phi^{(T+1)}(\Vec{x}^{(T+1)}) = \Vec{0}$. Thus, by \Cref{lemma:QG} it follows that
\begin{align*}
    - \Phi^{(T+1)}(\Vec{g}^{(T+1)}) + \Phi^{(T+1)}(\Vec{x}^{(T+1)}) \geq \omega(\|\Vec{x}^{(T+1)} - \Vec{g}^{(T+1)}\|_{\Vec{x}^{(T+1)}}),
\end{align*}
since $-\Phi^{(T+1)}$ is self-concordant. Plugging this bound to \eqref{align:ind-second} implies that 
\begin{align}
    \sum_{t=1}^{T+1} &\left( \langle \Vec{g}^{(t)} - \Vec{x}^{(t)}, \Vec{m}^{(t)} \rangle - \langle \Vec{g}^{(t)}, \Vec{u}^{(t)} \rangle \right) \leq \notag \\
    &- \left\langle \Vec{g}^{(T+1)}, \Vec{m}^{(T+1)} + \sum_{t=1}^T \Vec{u}^{(t)} \right\rangle + \frac{\cR(\Vec{g}^{(T+1)})}{\eta} + \langle \Vec{g}^{(T+1)}, \Vec{m}^{(T+1)} \rangle - \langle \Vec{g}^{(T+1)}, \Vec{u}^{(T+1)} \rangle \notag \\ &- \frac{1}{\eta} \sum_{t=1}^{T+1} \left( \omega( \| \Vec{x}^{(t)} - \Vec{g}^{(t)}\|_{\Vec{x}^{(t)}}) + \omega( \|\Vec{x}^{(t)} - \Vec{g}^{(t-1)} \|_{\Vec{g}^{(t-1)}}) \right) \notag \\
    &= - \left\langle \Vec{g}^{(T+1)}, \sum_{t=1}^{T+1} \Vec{u}^{(t)} \right\rangle + \frac{\cR(\Vec{g}^{(T+1)})}{\eta} \notag - \frac{1}{\eta} \sum_{t=1}^{T+1} \left( \omega( \| \Vec{x}^{(t)} - \Vec{g}^{(t)}\|_{\Vec{x}^{(t)}}) + \omega( \|\Vec{x}^{(t)} - \Vec{g}^{(t-1)} \|_{\Vec{g}^{(t-1)}}) \right) \notag \\
    &\le - \left\langle \Vec{x}^{*}, \sum_{t=1}^{T+1} \Vec{u}^{(t)} \right\rangle + \frac{\cR(\Vec{x}^{*})}{\eta} \notag - \frac{1}{\eta} \sum_{t=1}^{T+1} \left( \omega( \| \Vec{x}^{(t)} - \Vec{g}^{(t)}\|_{\Vec{x}^{(t)}}) + \omega( \|\Vec{x}^{(t)} - \Vec{g}^{(t-1)} \|_{\Vec{g}^{(t-1)}}) \right), \notag
\end{align}
for any $\Vec{x}^* \in \inter(\cX)$, where the last inequality follows since $\Psi^{(T+1)}(\Vec{g}^{(T+1)}) \geq \Psi^{(T+1)}(\Vec{x}^*)$, for any $\Vec{x}^* \in \inter(\cX)$, by definition of $\Vec{g}^{(T+1)}$. This establishes the inductive step, completing the proof of the theorem.
\end{proof}

Next, to cast \Cref{theorem:prel-rvu} in the form of an $\rvu$ bound (in the sense of~\citep{Syrgkanis15:Fast}), we establish the stability of the iterates as formalized below.

\begin{lemma}[Stability]
    \label{lemma:stability}
Let $\eta > 0$ be such that $\eta \| \Vec{u}^{(t)} - \Vec{m}^{(t)} \|_{*, \Vec{x}^{(t)}} \leq \frac{1}{2}$ and $\eta \| \Vec{m}^{(t)}\|_{*, \Vec{g}^{(t-1)}} \leq \frac{1}{2}$, for all $t \in \range{T}$. Then, for any $t \in \range{T}$,
\begin{align*}
    \| \Vec{x}^{(t)} - \Vec{g}^{(t)} \|_{\Vec{x}^{(t)}} \leq 2 \eta \|\Vec{u}^{(t)} - \Vec{m}^{(t)} \|_{*, \Vec{x}^{(t)}}; \\
    \| \Vec{x}^{(t)} - \Vec{g}^{(t-1)} \|_{\Vec{g}^{(t-1)}} \leq 2 \eta \| \Vec{m}^{(t)} \|_{*, \Vec{g}^{(t-1)}}.
\end{align*}
\end{lemma}

\begin{proof}
Fix any $t \in \range{T}$. We observe that $\| \Vec{x}^{(t)} - \Vec{g}^{(t)}\|_{\Vec{x}^{(t)}} = \| \Vec{x}^{(t)} - \argmin (-\Psi^{(t)}) \|_{\Vec{x}^{(t)}}$, by definition of \eqref{eq:BTL}. Further, we have that $\Psi^{(t)}(\Vec{x}) = \Phi^{(t)}(\Vec{x}) + \eta \langle \Vec{x}, \Vec{u}^{(t)} - \Vec{m}^{(t)} \rangle $, implying that $\nabla \Psi^{(t)} = \nabla \Phi^{(t)} + \eta (\Vec{u}^{(t)} - \Vec{m}^{(t)})$. By the first-order optimaility condition of the optimization problem associated with \eqref{eq:OFTRL}, it follows that $\nabla \Phi^{(t)}(\Vec{x}^{(t)}) = 0$, in turn implying that $
\nabla \Psi^{(t)}(\Vec{x}^{(t)}) = \eta ( \Vec{u}^{(t)} - \Vec{m}^{(t)})$. As a result, we have shown that $\lambda(\Vec{x}^{(t)}, - \Psi^{(t)}) = \| \nabla \Psi^{(t)} (\Vec{x}^{(t)}) \|_{*, \Vec{x}^{(t)}} = \eta \| \Vec{u}^{(t)} - \Vec{m}^{(t)} \|_{*, \Vec{x}^{(t)}} \leq \frac{1}{2}$, by assumption. Thus, \Cref{lemma:stability-lam} implies that 
\begin{equation*}
    \| \Vec{x}^{(t)} - \Vec{g}^{(t)}\|_{\Vec{x}^{(t)}} = \| \Vec{x}^{(t)} - \argmin (-\Psi^{(t)}) \|_{\Vec{x}^{(t)}} \leq 2 \lambda(\Vec{x}^{(t)}, - \Psi^{(t)}) = 2 \eta \| \Vec{u}^{(t)} - \Vec{m}^{(t)} \|_{*, \Vec{x}^{(t)}},
\end{equation*}
concluding the first part of the claim. Similarly, we have that $\| \Vec{x}^{(t)} - \Vec{g}^{(t-1)} \|_{\Vec{g}^{(t-1)}} = \| \Vec{g}^{(t-1)} - \argmin (-\Phi^{(t)})\|_{\Vec{g}^{(t-1)}}$, by definition of \eqref{eq:OFTRL}. Further, we observe that $\Phi^{(t)}(\Vec{x}) = \Psi^{(t-1)}(\Vec{x}) + \eta \langle \Vec{x}, \Vec{m}^{(t)} \rangle$, implying that $\nabla \Phi^{(t)} = \nabla \Psi^{(t-1)} + \eta \Vec{m}^{(t)}$. Moreover, by the first-order optimality condition of the optimization problem associated with \eqref{eq:BTL}, we have that $\nabla \Psi^{(t-1)}(\Vec{g}^{(t-1)}) = \Vec{0}$. In turn, this implies that $\nabla \Phi^{(t)}(\Vec{g}^{(t-1)}) = \eta \Vec{m}^{(t)}$. As a result, we have shown that $\lambda(\Vec{g}^{(t-1)}, - \Phi^{(t)}) = \| \nabla \Phi^{(t)}(\Vec{g}^{(t-1)})\|_{*, \Vec{g}^{(t-1)}} = \eta \| \Vec{m}^{(t)} \|_{*, \Vec{g}^{(t-1)}} \leq \frac{1}{2}$, by assumption. Thus, \Cref{lemma:stability-lam} implies that 
\begin{equation*}
    \| \Vec{x}^{(t)} - \Vec{g}^{(t-1)} \|_{\Vec{g}^{(t-1)}} = \| \Vec{g}^{(t-1)} - \argmin (-\Phi^{(t)})\|_{\Vec{g}^{(t-1)}} \leq 2 \lambda(\Vec{g}^{(t-1)}, - \Phi^{(t)}) = 2 \eta \| \Vec{m}^{(t)}\|_{*, \Vec{g}^{(t-1)}}.
\end{equation*}
\end{proof}

We are now ready to establish \Cref{theorem:rvu}, the statement of which is recalled below.

\rvuself*

\begin{proof}
First, combining \Cref{theorem:prel-rvu} with the fact that $\|\Vec{x}^{(t)} - \Vec{g}^{(t)}\|_{\Vec{x}^{(t)}} \leq 2 \eta \| \Vec{u}^{(t)} - \Vec{m}^{(t)} \|_{*, \Vec{x}^{(t)}}$ (by \Cref{lemma:stability}) yields that 
\begin{equation*}
    \reg^T(\Vec{x}^*) \leq \frac{\cR(\Vec{x}^*)}{\eta} + 2\eta \sum_{t=1}^T \|\Vec{u}^{(t)} - \Vec{m}^{(t)}\|^2_{*, \Vec{x}^{(t)}} - \frac{1}{\eta} \sum_{t=1}^T \left( \omega( \| \Vec{x}^{(t)} - \Vec{g}^{(t)}\|_{\Vec{x}^{(t)}}) +  \omega(\|\Vec{x}^{(t)} - \Vec{g}^{(t-1)} \|_{\Vec{g}^{(t-1)}}) \right).
\end{equation*}
Further, it follows from \Cref{lemma:stability} that $\| \Vec{x}^{(t)} - \Vec{g}^{(t)}\|_{\Vec{x}^{(t)}} \leq 1$ and $\| \Vec{x}^{(t)} - \Vec{g}^{(t-1)}\|_{\Vec{g}^{(t-1)}} \leq 1$. Thus, \Cref{fact:omega-lb} implies that 
\begin{equation*}
    \reg^T(\Vec{x}^*) \leq \frac{\cR(\Vec{x}^*)}{\eta} + 2\eta \sum_{t=1}^T \|\Vec{u}^{(t)} - \Vec{m}^{(t)}\|^2_{*, \Vec{x}^{(t)}} - \frac{1}{4\eta} \sum_{t=1}^T \left( \| \Vec{x}^{(t)} - \Vec{g}^{(t)}\|^2_{\Vec{x}^{(t)}} +  \|\Vec{x}^{(t)} - \Vec{g}^{(t-1)} \|^2_{\Vec{g}^{(t-1)}} \right).
\end{equation*}
\end{proof}

For our purposes, it will be convenient to cast \Cref{theorem:rvu} in the following form, using the additional assumption that the Hessian $\nabla^2 \cR$ is stable.

\begin{corollary}
    \label{corollary:Hessian-smooth}
    Suppose that $\cR$ is a nondegenerate self-concordant function for $\inter(\cX)$ such that $\nabla^2 \cR(\tilx) \preceq 2 \nabla^2 \cR(\vec{x})$ for any $\vec{x}, \tilx \in \inter(\cX)$ with $\| \vec{x} - \tilx \|_{\tilx} \leq \frac{1}{4}$. Moreover, let $\eta > 0$ be such that $\eta \| \Vec{u}^{(t)} - \Vec{m}^{(t)} \|_{*, \Vec{x}^{(t)}} \leq \frac{1}{8}$ and $\eta \| \Vec{m}^{(t)}\|_{*, \Vec{g}^{(t-1)}} \leq \frac{1}{2}$ for all $t \in \range{T}$. Then, the regret of \eqref{eq:OFTRL} under any sequence of utilities $\Vec{u}^{(1)}, \dots, \Vec{u}^{(T)}$ can be bounded as
    \begin{equation*}
        \reg^T(\Vec{x}^*) \leq \frac{\cR(\Vec{x}^*)}{\eta} + 2\eta \sum_{t=1}^T \|\Vec{u}^{(t)} - \Vec{m}^{(t)}\|^2_{*, \Vec{x}^{(t)}} - \frac{1}{16 \eta} \sum_{t=1}^T \| \vec{x}^{(t)} - \vec{x}^{(t-1)} \|^2_{\vec{x}^{(t-1)}}.
    \end{equation*}
\end{corollary}

\begin{proof}
First, by \Cref{lemma:stability} we know that $\|\vec{x}^{(t-1)} - \vec{g}^{(t-1)} \|_{\vec{x}^{(t-1)}} \leq 2 \eta \| \vec{u}^{(t-1)} - \vec{m}^{(t-1)} \|_{*, \vec{x}^{(t-1)}} \leq \frac{1}{4}$, for any $t \in \N$. Thus, by assumption, it follows that $\nabla^2 \cR(\vec{x}^{(t-1)}) \preceq 2 \nabla^2 \cR(\vec{g}^{(t-1)})$, in turn implying that $\| \vec{x}^{(t)} - \vec{g}^{(t-1)} \|^2_{\vec{x}^{(t-1)}} \leq 2 \| \vec{x}^{(t)} - \vec{g}^{(t-1)} \|^2_{\vec{g}^{(t-1)}}$. Further, the triangle inequality for the norm $\|\cdot\|_{\Vec{x}^{(t-1)}}$ implies that 
\begin{align*}
    \| \Vec{x}^{(t)} - \Vec{x}^{(t-1)}\|^2_{\Vec{x}^{(t-1)}} &\leq 2 \| \Vec{x}^{(t)} - \Vec{g}^{(t-1)}\|^2_{\Vec{x}^{(t-1)}} + 2 \| \Vec{g}^{(t-1)} - \Vec{x}^{(t-1)}\|^2_{\Vec{x}^{(t-1)}} \\
    &\leq 4 \| \Vec{x}^{(t)} - \Vec{g}^{(t-1)}\|^2_{\Vec{g}^{(t-1)}} + 4 \| \Vec{x}^{(t-1)} - \Vec{g}^{(t-1)}\|^2_{\Vec{x}^{(t-1)}},
\end{align*}
where we used Young's inequality in the first line, and the fact that $\| \vec{x}^{(t)} - \vec{g}^{(t-1)} \|^2_{\vec{x}^{(t-1)}} \leq 2 \| \vec{x}^{(t)} - \vec{g}^{(t-1)} \|^2_{\vec{g}^{(t-1)}}$ in the second line. Thus, summing over all $t \in \range{T}$ yields that 
\begin{align*}
    \sum_{t=1}^T \| \Vec{x}^{(t)} - \Vec{x}^{(t-1)}\|^2_{\Vec{x}^{(t-1)}} &\leq 4 \sum_{t=1}^T \| \Vec{x}^{(t)} - \Vec{g}^{(t-1)}\|^2_{\Vec{g}^{(t-1)}} + 4 \sum_{t=1}^T \| \Vec{x}^{(t-1)} - \Vec{g}^{(t-1)}\|^2_{\Vec{x}^{(t-1)}} \\
    &\leq 4 \sum_{t=1}^T \| \Vec{x}^{(t)} - \Vec{g}^{(t-1)}\|^2_{\Vec{g}^{(t-1)}} + 4 \sum_{t=1}^T \| \Vec{x}^{(t)} - \Vec{g}^{(t)}\|^2_{\Vec{x}^{(t)}},
\end{align*}
since $\vec{x}^{(0)} = \vec{g}^{(0)}$. Finally, plugging this bound to \Cref{theorem:rvu} concludes the proof.
\end{proof}

\begin{corollary}[Stability of the Iterates]
    \label{corollary:stability}
    Suppose that $\cR$ is a self-concordant function for $\inter(\cX)$ such that $\nabla^2 \cR(\tilx) \preceq 2 \nabla^2 \cR(\vec{x})$ for any $\vec{x}, \tilx \in \inter(\cX)$ with $\| \vec{x} - \tilx \|_{\tilx} \leq \frac{1}{4}$. Moreover, let $\eta > 0$ be such that $\eta \| \Vec{u}^{(t)} - \Vec{m}^{(t)} \|_{*, \Vec{x}^{(t)}} \leq \frac{1}{8}$ and $\eta \| \Vec{m}^{(t)}\|_{*, \Vec{g}^{(t-1)}} \leq \frac{1}{2}$ for all $t \in \range{T}$. Then, 
    \begin{equation*}
        \| \vec{x}^{(t)} - \vec{x}^{(t-1)} \|_{\vec{x}^{(t-1)}} \le 4\eta \| \vec{m}^{(t)} \|_{*, \vec{g}^{(t-1)}} + 2 \eta \| \vec{u}^{(t-1)} - \vec{m}^{(t-1)} \|_{*, \vec{x}^{(t-1)}}.
    \end{equation*}
\end{corollary}

\begin{proof}
Similarly to the proof of \Cref{corollary:Hessian-smooth}, we obtain that
\begin{align*}
    \|\vec{x}^{(t)} - \vec{x}^{(t-1)} \|_{\vec{x}^{(t-1)}} &\leq \|\vec{x}^{(t)} - \vec{g}^{(t-1)} \|_{\vec{x}^{(t-1)}} + \|\vec{g}^{(t-1)} - \vec{x}^{(t-1)} \|_{\vec{x}^{(t-1)}} \\
    &\leq 2 \|\vec{x}^{(t)} - \vec{g}^{(t-1)} \|_{\vec{g}^{(t-1)}} + \|\vec{x}^{(t-1)} - \vec{g}^{(t-1)} \|_{\vec{x}^{(t-1)}} \\
    &\leq 4 \eta \|\vec{m}^{(t)} \|_{*, \vec{g}^{(t-1)}} + 2 \eta \|\vec{u}^{(t-1)} - \vec{m}^{(t-1)} \|_{*, \vec{x}^{(t-1)}}.
\end{align*}
\end{proof}

\subsection{Log-Barrier Regularizer on the Simplex}
\label{appendix:log-barrier}

Next, we instantiate our general $\rvu$ bound for the probability simplex. To this end, let us first point out that, leveraging \Cref{lemma:sum-bar} (and \Cref{lemma:comp}), we can construct a self-concordant barrier for any polytope defined by a set of inequalities $\mat{A} \Vec{x} \geq \Vec{b}$, for a matrix $\mat{A} \in \R^{k \times d}$ and a vector $\Vec{b} \in \R^k$, as pointed out below.

\begin{definition}[Log-Barrier Regularizer for Polytopes]
    \label{def:log-barrier}
    Consider any polytope defined by a set of inequalities $\mat{A} \vec{x} \geq \vec{b}$, for a matrix $\mat{A} \in \R^{k \times d}$ and a vector $\vec{b} \in \R^k$. The \emph{log-barrier} function $\cR$ is defined as \begin{equation}
    \label{eq:general-reg}
    \cR(\Vec{x}) \defeq - \sum_{\ind=1}^k \log (\mat{A}[\ind, :] \Vec{x} - \Vec{b}[\ind]).
\end{equation}
\end{definition}

Indeed, \Cref{lemma:sum-bar} implies that $\cR$ is a $k$-self-concordant barrier as it can be expressed as the sum of $k$ $1$-self-concordant barriers. Now let us focus on constructing a self-concordant barrier for the $(d-1)$-dimensional simplex $\Delta^d \defeq \left\{ \Vec{x} \in \R_{\geq 0}^{d} : \sum_{\ind=1}^{d} \Vec{x}[\ind] = 1 \right\}$. To address the fact that $\Delta^d$ has empty interior, we will restrict the problem to the domain $\simtrun \defeq  \left\{ \Vec{x} \in \R_{\geq 0}^{d-1} : \sum_{\ind=1}^{d-1} \Vec{x}[\ind] \leq 1 \right\}$. For notational convenience, we will also let $\Vec{x}[d] = 1 - \sum_{\ind=1}^{d-1} \Vec{x}[\ind]$. Thus, using the general \emph{log-barrier} regularizer for polytopes given in \eqref{eq:general-reg}, we arrive at the log-barrier regularizer for $\simtrun$:

\begin{equation}
    \label{eq:log-barrier}
    \cR(\Vec{x}) \defeq - \sum_{r=1}^{d-1} \log (\Vec{x}[r]) - \log\left( 1 - \sum_{r=1}^{d-1} \Vec{x}[r] \right).
\end{equation}

Naturally, $\cR$ is a $d$-self-concordant barrier since it can be expressed as the sum of $d$ $1$-self-concordant barriers. It is important to stress that the regularizer given in \eqref{eq:log-barrier} takes as input a $(d-1)$-dimensional vector. To reconcile this with the fact that the regret minimizer should receive a $d$-dimensional utility vector $\Vec{u} \in \R^d$, in the sequel we will use a simple transformation of the observed utilities (while preserving the incurred regret). But first, let us also introduce an auxiliary regularizer for the purpose of our analysis; namely, 
\begin{equation}
    \label{eq:aux-reg}
    \auxR(\Vec{x}) \defeq - \sum_{\ind = 1}^d \log \Vec{x}[\ind].
\end{equation}
We are going to relate the local norm induced by the log-barrier \eqref{eq:log-barrier} to that induced by the auxiliary regularizer \eqref{eq:aux-reg}. First, we characterize the primal local norm induced by $\cR$ and $\auxR$.
\begin{claim}
    \label{claim:primal}
For any $\Vec{x}, \tilx \in \inter(\simtrun)$, 
\begin{equation*}
    \| \Vec{x} - \tilx\|^2_{\cR, \Vec{x}} = \sum_{\ind=1}^{d} \left( \frac{\Vec{x}[\ind] - \tilx[\ind]}{\Vec{x}[\ind]} \right)^2. 
\end{equation*}
\end{claim}

\begin{proof}
Let us first compute the Hessian of $\cR$. A direct calculation gives that for $\ind \in \range{d-1}$,
\begin{equation*}
    \frac{\partial^2 \cR }{\partial \Vec{x}[\ind]^2} = \frac{1}{(\Vec{x}[\ind])^2} + \frac{1}{\left( 1 - \sum_{\ind=1}^{d-1} \Vec{x}[\ind] \right)^2} = \frac{1}{(\Vec{x}[\ind])^2} + \frac{1}{(\Vec{x}[d])^2},
\end{equation*}
where recall that $\Vec{x}[d] = 1 - \sum_{r=1}^{d-1} \Vec{x}[r]$ (by convention). Further, for $r' \neq r \in \range{d-1}$ we have that 
\begin{equation*}
    \frac{\partial^2 \cR }{\partial \Vec{x}[\ind] \partial \Vec{x}[\ind']} = \frac{\partial^2 \cR}{\partial \Vec{x}[\ind'] \partial \Vec{x}[\ind]} = \frac{1}{(\Vec{x}[d])^2}.
\end{equation*}
Thus, the Hessian of $\cR$ reads
\begin{equation}
    \label{eq:Hessian-R}
    \nabla^2 \cR = \diag \left( \frac{1}{(\Vec{x}[1])^2}, \dots, \frac{1}{(\Vec{x}[d-1])^2} \right) + \frac{1}{(\Vec{x}[d])^2} \Vec{1}_{d-1} \Vec{1}_{d-1}^\top.
\end{equation}
As a result,
\begin{align}
    \| \Vec{x} - \tilx \|^2_{\cR, \Vec{x}} &= (\Vec{x} - \tilx )^\top \diag \left( \frac{1}{(\Vec{x}[1])^2}, \dots, \frac{1}{(\Vec{x}[d-1])^2} \right) (\Vec{x} - \tilx ) + \frac{(\Vec{1}_{d-1}^\top (\Vec{x} - \tilx))^2}{(\Vec{x}[d])^2} \notag \\
    &= \sum_{\ind=1}^{d-1} \left( \frac{\Vec{x}[\ind] - \tilx[\ind]}{\Vec{x}[\ind]} \right)^2 + \left( \frac{\sum_{\ind=1}^{d-1} \Vec{x}[\ind] - \sum_{\ind=1}^{d-1} \tilx[\ind] }{\Vec{x}[d]} \right)^2 \notag \\
    &= \sum_{\ind=1}^{d-1} \left( \frac{\Vec{x}[\ind] - \tilx[\ind]}{\Vec{x}[\ind]} \right)^2 + \left( \frac{ \Vec{x}[d] - \tilx[d]}{\Vec{x}[d]} \right)^2 \notag \\
    &= \sum_{\ind=1}^{d} \left( \frac{\Vec{x}[\ind] - \tilx[\ind]}{\Vec{x}[\ind]} \right)^2. \notag
\end{align}
\end{proof}

Next, we characterize the dual norm induced by the regularizer $\cR$. To this end, let us first explain how the regret minimizer over the domain $\simtrun$ should operate. Upon observing a utility vector $\Vec{u} \in \R^d$, we construct the vector $\auxu \in \R^{d-1}$ so that $\auxu[\ind] = \Vec{u}[\ind] - \Vec{u}[d]$, for all $\ind \in \range{d-1}$. It is easy to see that the regret incurred is preserved through this transformation.
\begin{claim}
    \label{claim:dual}
For any $\auxu \in \R^{d-1}$ and $\Vec{x} \in \inter(\simtrun)$,
\begin{equation*}
    \| \auxu \|_{*, \cR, \Vec{x}} = \| \Vec{u} - \optcon \Vec{1}_d \|_{*, \auxR, \Vec{x}},
\end{equation*}
where $\optcon$ is the scalar that minimizes the norm in the right hand side.
\end{claim}

\begin{proof}
First, using the Sherman–Morrison formula we find that the inverse of the Hessian of $\cR$ given in \eqref{eq:Hessian-R} can be expressed as
\begin{equation*}
    (\nabla^2 \cR)^{-1} = \diag(\tilx[1], \dots, \tilx[d-1]) - \frac{1}{\sum_{r=1}^d (\Vec{x}[r])^2} \tilx \tilx^\top,
\end{equation*}
where $\tilx \defeq ((\Vec{x}[1])^2, \dots, (\Vec{x}[d-1])^2)$. Thus, by definition of $\auxu$ we have that
\begin{align}
    \| \auxu \|_{*, \cR, \Vec{x}} &= \sum_{\ind = 1}^{d-1} (\Vec{x}[\ind])^2 (\Vec{u}[\ind] - \Vec{u}[d] )^2 - \frac{( \sum_{\ind=1}^{d-1} (\Vec{x}[\ind])^2 (\Vec{u}[\ind] - \Vec{u}[d]))^2}{\sum_{\ind=1}^d (\Vec{x}[\ind])^2} \notag \\
    &= \sum_{\ind = 1}^{d} (\Vec{x}[\ind])^2 (\Vec{u}[\ind] - \Vec{u}[d] )^2 - \frac{( \sum_{\ind=1}^{d} (\Vec{x}[\ind])^2 \Vec{u}[\ind] - \Vec{u}[d] \sum_{\ind=1}^d (\Vec{x}[\ind])^2 )^2}{\sum_{\ind=1}^d (\Vec{x}[\ind])^2} \notag \\
    &= \sum_{\ind = 1}^{d} (\Vec{x}[\ind])^2 (\Vec{u}[\ind] - \Vec{u}[d] )^2 - \frac{(\sum_{\ind=1}^{d} (\Vec{x}[\ind])^2 \Vec{u}[\ind] )^2}{\sum_{\ind=1}^d (\Vec{x}[\ind])^2} + 2 \Vec{u}[d] \sum_{\ind=1}^d (\Vec{x}[\ind])^2 \Vec{u}[\ind] - (\Vec{u}[d])^2 \sum_{\ind=1}^d (\Vec{x}[\ind])^2 \notag \\
    &= \sum_{\ind = 1}^{d} (\Vec{x}[\ind] \Vec{u}[\ind] )^2 - \frac{(\sum_{\ind=1}^{d} (\Vec{x}[\ind])^2 \Vec{u}[\ind] )^2}{\sum_{\ind=1}^d (\Vec{x}[\ind])^2}, \label{align:cont}
\end{align}
by simple algebraic calculations. Now let us define the scalar $c^*$ as
\begin{equation*}
    \optcon \defeq \frac{\sum_{\ind=1}^d (\Vec{x}[\ind])^2 \Vec{u}[\ind]}{\sum_{\ind=1}^d(\Vec{x}[\ind])^2}.
\end{equation*}
Then, continuing from \eqref{align:cont},
\begin{align}
    \| \auxu \|_{*, \cR, \Vec{x}} &= \sum_{\ind=1}^d (\Vec{x}[\ind])^2 \left( (\Vec{u}[\ind])^2 - 2 \left( \frac{\sum_{\ind'=1}^d (\Vec{x}[\ind'])^2 \Vec{u}[\ind']}{\sum_{\ind'=1}^d(\Vec{x}[\ind'])^2} \right) \Vec{u}[\ind] + \left( \frac{\sum_{\ind'=1}^d (\Vec{x}[\ind'])^2 \Vec{u}[\ind']}{\sum_{\ind'=1}^d(\Vec{x}[\ind'])^2} \right)^2 \right) \notag \\
    &= \sum_{\ind=1}^r (\Vec{x}[\ind])^2 \left( \Vec{u}[\ind] - \optcon \right)^2 = \| \Vec{u} - \optcon \Vec{1}_d \|_{*, \auxR, \Vec{x}}. \label{align:opt-quadr}
\end{align}
But, it is easy to see that $c^*$ is the minimizer of \eqref{align:opt-quadr}. This concludes the proof.
\end{proof}

An analogous argument shows that $\| \auxu^{(t)} - \auxu^{(t-1)}\|_{*, \cR, \Vec{x}} = \| \vec{u}^{(t)} - \vec{u}^{(t-1)} - \optcon \vec{1}_d \|_{*, \auxR, \Vec{x}} \le \| \vec{u}^{(t)} - \vec{u}^{(t-1)}\|_{*, \auxR, \Vec{x}}$. Finally, combining \Cref{claim:primal} and \Cref{claim:dual} with \Cref{theorem:rvu} and \Cref{corollary:Hessian-smooth} directly leads to the $\rvu$ bound of \Cref{corollary:rvu-simplex}.

%% file: text/appendix_swap.tex
\section{Omitted Proofs from \texorpdfstring{\Cref{section:swap}}{Section 3}}
\label{appendix:swap}

In this section we provide the omitted proofs from \Cref{section:swap}. We begin by summarizing the construction of \citet{Blum07:From} in \Cref{algo:swap-BM}. We point out that a regret minimization algorithm $\regmin$ is modeled as a black box which interacts with its environment via the following two subroutines.

\begin{itemize}
    \item[(i)] $\regmin.\nextstr()$: $\regmin$ returns the next strategy of the learner;
    \item[(ii)] $\regmin.\obsut(\vec{u})$: $\regmin$ receives as feedback from the environment a utility vector $\vec{u}$, and may adapt its internal state accordingly.
\end{itemize}

\begin{algorithm}[H]
    \SetAlgoLined
    \DontPrintSemicolon
    \KwIn{A set of external regret minimizers $\left\{ \regmin_{a} \right\}_{a \in \cA}$, each for the simplex $\Delta(\cA)$}    
    \BlankLine
    \SetInd{2.3mm}{2.3mm}
    \Fn{$\nextstr()$}{
        $\mat{Q}^{(t)} \gets \mat{0}\in \R^{|\cA| \times |\cA|}$\;
        \For{$a \in \cA$}{
        $\mat{Q}^{(t)}[a,\cdot] \leftarrow \regmin_{a}.\nextstr()$\;
        }
        $\vec{x}^{(t)}\gets\fp(\mat{Q}^{(t)})$\;
        \textbf{return} $\vec{x}^{(t)}$
    }
    \Hline{}
    \Fn{$\obsut(\vec{u}^{(t)})$}{
    \For{$a \in \cA$}{
        $\regmin_{a}.\obsut(\vec{x}\^{t}[a]\, \vec{u}^{(t)})$
    }
    }
\caption{\citet{Blum07:From}}
\label{algo:swap-BM}
\end{algorithm}

We start with the proof of \Cref{lemma:gamma-term}. To this end, we first apply \Cref{corollary:rvu-simplex} for each individual regret minimizer $\regmin_a$, leading to the following guarantee for $\eta \leq \frac{1}{16}$.
\begin{equation}
    \label{eq:external_rega}
    \reg_a^T(\Vec{x}_a^*) \leq \frac{\cR(\Vec{x}_a^*)}{\eta} + 2 \eta \sum_{t=1}^T \| \Vec{u}^{(t)} \vec{x}^{(t)}[a] - \Vec{u}^{(t-1)} \vec{x}^{(t-1)}[a] \|^2_{*, \Vec{x}_a^{(t)}} - \frac{1}{16\eta} \sum_{t=1}^T \| \Vec{x}_a^{(t)} - \Vec{x}_a^{(t-1)} \|^2_{\Vec{x}_a^{(t-1)}},
\end{equation}
for any $\vec{x}_a^* \in \relint(\cA)$; it is assumed that each regret minimizer $\regmin_a$ is employing the same learning rate $\eta > 0$. Next, the triangle inequality along with Young's inequality imply that 
\begin{equation*}
    \| \Vec{u}^{(t)} \vec{x}^{(t)}[a] - \Vec{u}^{(t-1)} \vec{x}^{(t-1)}[a] \|^2_{*, \Vec{x}_a^{(t)}} \leq 2 (\vec{x}^{(t)}[a])^2 \| \vec{u}^{(t)} - \vec{u}^{(t-1)} \|^2_{*, \vec{x}_a^{(t)}} + 2 (\vec{x}^{(t)}[a] - \vec{x}^{(t-1)}[a])^2 \| \vec{u}^{(t-1)}\|^2_{*, \vec{x}_a^{(t)}},
\end{equation*}
for any $a \in \cA$. Summing this inequality over all $a \in \cA$ yields that 
\begin{equation}
    \label{eq:bound-beta}
    \sum_{a \in \cA} \| \Vec{u}^{(t)} \vec{x}^{(t)}[a] - \Vec{u}^{(t-1)} \vec{x}^{(t-1)}[a] \|^2_{*, \Vec{x}_a^{(t)}} \le 2 \| \vec{u}^{(t)} - \vec{u}^{(t-1)} \|^2_\infty + 2 \| \vec{x}^{(t)} - \vec{x}^{(t-1)} \|_2^2.
\end{equation}
Next, let us address the diameter term in \eqref{eq:external_rega}. Let $\vec{x}_c \defeq \argmin_{\vec{x}} \cR(\vec{x})$, so that $\cR(\vec{x}_c) = 0$. If $\pi(\vec{x}_a^*; \vec{x}_c) \leq 1 - \frac{1}{T}$, then, by \Cref{theorem:diameter},
\begin{equation*}
    \cR(\vec{x}^*_a) \leq |\cA| \log \left( \frac{1}{1 - \pi(\vec{x}_a^*; \vec{x}_c)} \right) \leq  m \log T,
\end{equation*}
where we used the notation $m \defeq |\cA|$. Otherwise, we define $\tilx_a^* \defeq (1 - 1/T) \vec{x}_a^* + (1/T) \vec{x}_c$, and we observe that 
\begin{equation*}
    \reg_a^T(\vec{x}_a^*) \leq \reg_a^T(\tilx_a^*) + \sum_{t=1}^T \langle \vec{x}_a^* - \tilx^*_a, \vec{x}^{(t)}[a] \vec{u}^{(t)} \rangle \leq \reg_a^T(\tilx_a^*) + \frac{2}{T} \sum_{t=1}^T \vec{x}^{(t)}[a] \| \vec{u}^{(t)} \|_\infty.
\end{equation*}
Thus, from \eqref{eq:external_rega} we conclude that $\reg_a^T$ is upper bounded by
\begin{equation}
    \label{eq:external_rega+}
    \frac{m \log T}{\eta} + \frac{2}{T} \sum_{t=1}^T \vec{x}^{(t)}[a] + 2 \eta \sum_{t=1}^T \| \Vec{u}^{(t)} \vec{x}^{(t)}[a] - \Vec{u}^{(t-1)} \vec{x}^{(t-1)}[a] \|^2_{*, \Vec{x}_a^{(t)}} - \frac{1}{16 \eta} \sum_{t=1}^T \| \Vec{x}_a^{(t)} - \Vec{x}_a^{(t-1)} \|^2_{\Vec{x}_a^{(t-1)}},
\end{equation}
since $\| \vec{u}^{(t)}\|_\infty \leq 1$. Next, we will use the fact that the log-barrier regularizer guarantees \emph{multiplicative stability}, in the following formal sense.

\begin{corollary}[Multiplicative Stability]
    \label{corollary:mul-stable}
    In the setting of \Cref{corollary:stability}, suppose that $\|\vec{u}^{(t)}\|_{\infty}, \|\vec{m}^{(t)}\|_{\infty} \leq 1$ for all $t \in \range{T}$. If $\eta \leq \frac{1}{16}$, then for $2 \leq t \leq T$,
    \begin{equation*}
        \sqrt{\sum_{\ind=1}^d \left( 1 - \frac{\vec{x}^{(t)}[\ind]}{\vec{x}^{(t-1)}[\ind]} \right)^2} \leq 6 \eta \| \vec{u}^{(t-1)} \|_\infty + 2 \eta \| \vec{u}^{(t-2)}\|_\infty.
    \end{equation*}
\end{corollary}
\begin{proof}
The claim follows directly from \Cref{corollary:stability} with $\vec{m}^{(t)} = \vec{u}^{(t-1)}$, using the fact that $\| \vec{u} \|_{*, \vec{x}^{(t)}} \leq \|\vec{u}\|_\infty$.
\end{proof}
Now let 
\begin{equation*}
    \mul^{(t)}_a \defeq \max_{a' \in \cA} \left| 1 - \frac{\vec{x}_a^{(t)}[a']}{\vec{x}_a^{(t-1)}[a']} \right|,
\end{equation*}
for each $a \in \cA$. \Cref{corollary:mul-stable} implies that 
\begin{equation*}
    \mul^{(t)}_a \leq 6 \eta \| \vec{u}^{(t-1)} \vec{x}^{(t-1)}[a] \|_\infty + 2 \eta \| \vec{u}^{(t-2)} \vec{x}^{(t-2)}[a] \|_\infty = 6 \eta \vec{x}^{(t-1)}[a] \| \vec{u}^{(t-1)} \|_\infty + 2 \eta \vec{x}^{(t-2)}[a] \| \vec{u}^{(t-2)} \|_\infty.
\end{equation*}
Thus, summing over all $a \in \cA$ yields that 
\begin{equation}
    \label{sum-of-etas}
    \sum_{a \in \cA} \mul^{(t)}_a \leq 6 \eta \sum_{a \in \cA} \vec{x}^{(t-1)}[a] \| \vec{u}^{(t-1)} \|_\infty + 2 \eta \sum_{a \in \cA} \vec{x}^{(t-2)}[a] \|\vec{u}^{(t-2)}\|_\infty \le 8 \eta,
\end{equation}
for $t \geq 2$, where we used that $\vec{x}^{(t-1)}, \vec{x}^{(t-2)} \in \Delta(\cA)$, as well as the normalization assumption $\|\vec{u}\|_\infty \leq 1$; it is also immediate to see that $\sum_{a \in \cA} \mul_a^{(1)} \leq 8 \eta$.

For the proof of \Cref{lemma:gamma-term} we will require the Markov chain tree theorem. In particular, consider an $m$-node ergodic (\emph{i.e.}, aperiodic and irreducible) Markov chain represented through a row-stochastic matrix $\mat{Q}$. The Markov chain tree theorem establishes a closed-form solution for the (unique) stationary distribution $\vec{\pi}$; that is, the vector $\vec{\pi} \in \Delta^m$ for which $\vec{\pi}^\top \mat{Q} = \vec{\pi}^\top$. To this end, we formalize the notion of a directed tree.

\begin{definition}[Directed Tree]
    A directed graph $\tree = (V, E)$ is a \emph{directed tree} rooted at node $a$ if (i) it containts no (directed) cycles; (ii) every node $V \setminus \{a\}$ has exactly one outgoing edge; and (iii) the root node $a$ has no outgoing edges. 
\end{definition}

We will denote with $\treeset_a$ the set of all possible directed $m$-node trees rooted at node $a$. Finally, before we state the Markov chain tree theorem, we let $\Sigma_a$ be defined as
\begin{equation}
    \label{eq:Sigma-a}
    \Sigma_a = \sum_{\tree \in \treeset_a} \prod_{(u, v) \in E(\tree)} \mat{Q}[u, v].
\end{equation}

\begin{theorem}[Markov Chain Tree Theorem; \emph{e.g.}, \citep{Anantharam89:A}]
    \label{theorem:MCTT}
    The stationary distribution $\vec{\pi} \in \Delta^m$ of an $m$-state ergodic markov chain with row-stochastic transition matrix $\mat{Q}$ is such that 
    \begin{equation*}
        \vec{\pi}[a] = \frac{\Sigma_a}{\Sigma},
    \end{equation*}
    where $\Sigma \defeq \sum_{a} \Sigma_a$, and each $\Sigma_a$ is defined as in \eqref{eq:Sigma-a}.
\end{theorem}

We are now ready to prove \Cref{lemma:gamma-term}.

\gammaterm*

\begin{proof}
Consider any $t \in \N$. From the Markov chain tree theorem (\Cref{theorem:MCTT}) we know that 
\begin{equation*}
    \vec{x}[a] = \frac{\Sigma_a}{\Sigma}, \quad \forall a \in \cA,
\end{equation*}
where $\Sigma_a \defeq \sum_{\tree \in \treeset_a} \prod_{(u,v) \in E(\tree)} \mat{Q}[u, v]$ and $\Sigma = \sum_{a \in \cA} \Sigma_a$. Fix some action $a \in \cA$ and a directed tree $\tree \in \treeset_a$ rooted at node $a$. Then, 
\begin{equation*}
    \prod_{(u,v) \in E(\tree)} \mat{Q}^{(t)}[u, v] = \prod_{(u,v) \in E(\tree)} \vec{x}_u^{(t)}[v] \leq \prod_{(u, v) \in E(\tree)} (1 + \mul^{(t)}_u) \vec{x}_u^{(t-1)}[v],
\end{equation*}
where we used the fact that 
\begin{equation*}
    \mul_u^{(t)} \geq \frac{\vec{x}_u^{(t)}[v]}{\vec{x}_u^{(t-1)}[v]} - 1 \implies \vec{x}_u^{(t)}[v] \leq (1 + \mul_u^{(t)}) \vec{x}_u^{(t-1)}[v] ,
\end{equation*}
Thus, 
\begin{equation*}
    \prod_{(u,v) \in E(\tree)} \mat{Q}^{(t)}[u,v] \leq \prod_{u \neq a} (1 + \mul^{(t)}_u) \prod_{(u,v) \in E(\tree)} \mat{Q}^{(t-1)}[u,v],
\end{equation*}
where we used the fact that $\tree$ is a directed tree rooted at $a$. Thus, summing over all $\tree \in \treeset_a$ yields that 
\begin{align}
    \Sigma_a^{(t)} = \sum_{\tree \in \treeset_a} \prod_{(u,v) \in E(\tree)} \mat{Q}^{(t)}[u,v] &\leq \Sigma_a^{(t-1)} \prod_{a' \in \cA} (1 + \mul^{(t)}_{a'}) \notag \\
    &\leq \Sigma_a^{(t-1)} \exp \left\{ \sum_{a' \in \cA} \mul^{(t)}_{a'} \right\}. \label{align:Sigmaa-up}
\end{align}
This also implies that
\begin{equation}
    \label{eq:Sigma-up}
    \Sigma^{(t)} = \sum_{a \in \cA} \Sigma_a^{(t)} \leq \exp \left\{ \sum_{a' \in \cA} \mul^{(t)}_{a'} \right\} \sum_{a \in \cA} \Sigma_a^{(t-1)} = \exp \left\{ \sum_{a' \in \cA} \mul^{(t)}_{a'} \right\} \Sigma^{(t-1)}. 
\end{equation}
Similarly, 
\begin{equation*}
    \prod_{(u,v) \in E(\tree)} \mat{Q}^{(t)}[u, v] = \prod_{(u,v) \in E(\tree)} \vec{x}_u^{(t)}[v] \geq \prod_{(u, v) \in E(\tree)} (1 - \mul^{(t)}_u) \vec{x}_u^{(t-1)}[v],
\end{equation*}
where we used the fact that 
\begin{equation*}
    \mul_u^{(t)} \geq 1 - \frac{\vec{x}_u^{(t)}[v]}{\vec{x}_u^{(t-1)}[v]} \implies \vec{x}_u^{(t)}[v] \geq (1 - \mul^{(t)}_u) \vec{x}_u^{(t-1)}[v].
\end{equation*}
Thus, summing over all $\tree \in \treeset_a$ implies that 
\begin{align}
    \Sigma_a^{(t)} = \sum_{\tree \in \treeset_a} \prod_{(u,v) \in E(\tree)} \mat{Q}^{(t)}[u,v] &\geq \Sigma_a^{(t-1)} \prod_{a' \in \cA} (1 - \mul^{(t)}_{a'}) \notag \\
    &\geq \Sigma_a^{(t-1)} \exp \left\{ - 2 \sum_{a' \in \cA} \mul^{(t)}_{a'} \right\}, \label{align:Sigmaa-lo} 
\end{align}
where we used the inequality $1 - x \geq e^{-2x}$, for all $x \in [0, \frac{1}{2}]$, applicable since (by \eqref{sum-of-etas}) $\sum_{a' \in \cA} \mul_{a'}^{(t)} \leq 8 \eta \leq \frac{1}{2}$ for $\eta \leq \frac{1}{16}$. This also implies that 
\begin{equation}
    \label{eq:Sigma-lo}
    \Sigma^{(t)} = \sum_{a \in \cA} \Sigma_a^{(t)} \geq \exp \left\{ - 2 \sum_{a' \in \cA} \mul^{(t)}_{a'} \right\} \sum_{a \in \cA} \Sigma_a^{(t-1)} = \exp \left\{ - 2 \sum_{a' \in \cA} \mul^{(t)}_{a'} \right\} \Sigma^{(t-1)}.
\end{equation}
As a result, from \eqref{align:Sigmaa-up} and \eqref{eq:Sigma-lo} it follows that for any $a \in \cA$,
\begin{align*}
    \frac{\Sigma_a^{(t)}}{\Sigma^{(t)}} - \frac{\Sigma_a^{(t-1)}}{\Sigma^{(t-1)}} \leq \frac{\Sigma_a^{(t-1)} \exp \left\{ \sum_{a' \in \cA} \mul^{(t)}_{a'} \right\} }{\Sigma^{(t-1)} \exp \left\{ - 2 \sum_{a' \in \cA} \mul^{(t)}_{a'} \right\} } - \frac{\Sigma_a^{(t-1)}}{\Sigma^{(t-1)}} &= \frac{\Sigma_a^{(t-1)}}{\Sigma^{(t-1)}} \left(\exp \left\{  3 \sum_{a' \in \cA} \mul^{(t)}_{a'} \right\} - 1 \right) \\
    &\leq \frac{\Sigma_a^{(t-1)}}{\Sigma^{(t-1)}} \left( 8 \sum_{a' \in \cA} \mul^{(t)}_{a'} \right),
\end{align*}
where we used the inequality $e^x - 1 \leq \frac{8}{3} x$ for all $x \in [0,\frac{3}{2}]$, applicable since $\sum_{a' \in \cA} \mul_{a'}^{(t)} \leq \frac{1}{2}$. Similarly, \eqref{align:Sigmaa-lo} and \eqref{eq:Sigma-up} imply that for any $a \in \cA$,
\begin{align*}
    \frac{\Sigma_a^{(t-1)}}{\Sigma^{(t-1)}} - \frac{\Sigma_a^{(t)}}{\Sigma^{(t)}} \leq \frac{\Sigma_a^{(t-1)}}{\Sigma^{(t-1)}} - \frac{\Sigma_a^{(t-1)} \exp \left\{ - 2 \sum_{a' \in \cA} \mul^{(t)}_{a'} \right\} }{\Sigma^{(t-1)} \exp \left\{ \sum_{a' \in \cA} \mul^{(t)}_{a'} \right\} } &= \frac{\Sigma_a^{(t-1)}}{\Sigma^{(t-1)}} \left(1 - \exp \left\{ - 3 \sum_{a' \in \cA} \mul^{(t)}_{a'} \right\}\right) \\
    &\leq \frac{\Sigma_a^{(t-1)}}{\Sigma^{(t-1)}} \left( 3 \sum_{a' \in \cA} \mul^{(t)}_{a'} \right).
\end{align*}
As a result, we have established that 
\begin{equation*}
    \left| \vec{x}^{(t)}[a] - \vec{x}^{(t-1)}[a] \right| = \left| \frac{\Sigma_a^{(t)}}{\Sigma^{(t)}} - \frac{\Sigma_a^{(t-1)}}{\Sigma^{(t-1)}} \right| \leq 8 \frac{\Sigma_a^{(t-1)}}{\Sigma^{(t-1)}} \sum_{a' \in \cA} \mul^{(t)}_{a'} = 8 \vec{x}^{(t-1)}[a] \sum_{a' \in \cA} \mul^{(t)}_{a'},
\end{equation*}
in turn implying that 
\begin{equation}
    \label{eq:penul}
    \| \vec{x}^{(t)} - \vec{x}^{(t-1)} \|_1 \leq 8 \left( \sum_{a' \in \cA} \mul^{(t)}_{a'} \right) \left( \sum_{a \in \cA} \vec{x}^{(t-1)}[a] \right) = 8 \sum_{a' \in \cA} \mul^{(t)}_{a'},
\end{equation}
since $\vec{x}^{(t-1)} \in \Delta(\cA)$. Thus,
\begin{equation*}
    \| \vec{x}^{(t)} - \vec{x}^{(t-1)}\|_1^2 \leq 64 \left( \sum_{a \in \cA} \mul^{(t)}_{a} \right)^2 \leq 64 |\cA| \sum_{a \in \cA} \left( \mul_a^{(t)} \right)^2,
\end{equation*}
by Jensen's inequality. Finally, 
\begin{equation*}
    \left( \mul_a^{(t)} \right)^2 = \max_{a' \in \cA} \left( 1 - \frac{\vec{x}_a^{(t)}[a']}{\vec{x}_a^{(t-1)}[a']} \right)^2 \leq \sum_{a' \in \cA} \left( 1 - \frac{\vec{x}_a^{(t)}[a']}{\vec{x}_a^{(t-1)}[a']} \right)^2 = \| \vec{x}_a^{(t)} - \vec{x}_a^{(t-1)} \|^2_{\vec{x}_a^{(t-1)}},
\end{equation*}
and combining this bound with \eqref{eq:penul} concludes the proof.
\end{proof}

\rvuswap*

\begin{proof}
Combining \eqref{eq:external_rega+}, \eqref{eq:bound-beta}, \Cref{theorem:swap-external}, and \Cref{lemma:gamma-term} implies that $\swapreg_i^T$ is upper bounded by
\begin{equation*}
    \frac{2 m^2 \log T}{\eta} + 4\eta \sum_{t=1}^T \| \vec{u}^{(t)} - \vec{u}^{(t-1)}\|^2_\infty + 4 \eta \sum_{t=1}^T \| \vec{x}^{(t)} - \vec{x}^{(t-1)} \|_2^2  - \frac{1}{1024 m \eta} \sum_{t=1}^T  \| \Vec{x}^{(t)} - \Vec{x}^{(t-1)}\|_1^2.
\end{equation*}
Further, for $\eta \leq \frac{1}{128 \sqrt{m}}$ it follows that 
\begin{equation*}
    4\eta \|\vec{x}^{(t)} - \vec{x}^{(t-1)}\|_2^2 \leq 4\eta \|\vec{x}^{(t)} - \vec{x}^{(t-1)}\|_1^2 \leq \frac{1}{2048 m \eta} \|\vec{x}^{(t)} - \vec{x}^{(t-1)}\|_1^2.
\end{equation*}
In turn, this implies that
\begin{equation*}
    \swapreg_i^T \leq \frac{2 m^2 \log T }{\eta} + 4\eta \sum_{t=1}^T \| \vec{u}^{(t)} - \vec{u}^{(t-1)}\|^2_\infty - \frac{1}{2048 m \eta} \sum_{t=1}^T  \| \Vec{x}^{(t)} - \Vec{x}^{(t-1)}\|_1^2,
\end{equation*}
concluding the proof.
\end{proof}

\optswap*

\begin{proof}
By \Cref{theorem:rvu-swap} and \Cref{theorem:log-bounded_trajectories},
\begin{align*}
    \swapreg_i^T &\leq \frac{2m_i^2 \log T}{\eta} + 4 \eta (n-1) \sum_{j \neq i} \sum_{t=1}^T \|\vec{x}_j^{(t)} - \vec{x}_j^{(t-1)} \|_1^2 \\
    &\leq \frac{2m_i^2 \log T}{\eta} + 32768 \eta (n-1) \max_{j \in \range{n}} \{ m_j \} \sum_{j=1}^n m_j^{2} \log T \\
    &= 256 \max_{j \in \range{n}} \{ \sqrt{m_j} \} \left( (n-1) m_i^{2} + \sum_{j=1}^n m_j^{2} \right) \log T.
\end{align*}
\end{proof}

\adverob*

\begin{proof}
Each player $i \in \range{n}$ initially follows the $\bmoftrl$ dynamics with learning rate $\eta = \frac{1}{128 (n-1) \max_{j \in \range{n}} \{\sqrt{m_j}\} } $. Next, player $i$ keeps track of the quantity $\sum_{\tau=1}^t \|\vec{u}_i^{(\tau)} - \vec{u}_i^{(\tau-1)} \|_{\infty}^2$. If for all $2 \leq t \leq T$ it holds that
\begin{equation}
    \label{eq:adv-cond}
    \sum_{\tau=1}^t \|\vec{u}_i^{(\tau)} - \vec{u}_i^{(\tau-1)} \|_{\infty}^2 \leq 8192 (n-1) \max_{j \in \range{n}} \{ m_j \} \sum_{j = 1}^n m_j^{2} \log t,
\end{equation}
then the swap regret of player $i \in \range{n}$ enjoys the guarantee of \Cref{corollary:near-opt-swap}, as follows directly from \Cref{theorem:rvu-swap}. In particular, \eqref{eq:adv-cond} will hold as long as all players follow the prescribed dynamics, by virtue of \Cref{theorem:log-bounded_trajectories}.  Otherwise, let $t \geq 2$ be the \emph{first} iteration for which \eqref{eq:adv-cond} is violated. The overall swap regret accumulated up to time $t-1$ is at most the guarantee of \Cref{corollary:near-opt-swap}, as follows directly from \Cref{theorem:rvu-swap}, while the swap regret at time $t$ is at most $2$ since $\|\vec{u}_i^{(t)}\|_{\infty} \leq 1$. Next, the player switches to BM-MWU with learning rate $\eta = \sqrt{\frac{m_i \log m_i }{T}}$. Thereafter, the accumulated swap regret will be bounded by $2 \sqrt{m_i \log m_i T}$. This completes the proof.
\end{proof}

Finally, we conclude this section with a refinement for games with a large number of players. In particular, we will assume that the utility of each player only depends on the actions of a small number of other players (Item \ref{item:u1}), and that each player's actions only affect the utility of a small number of other players (Item \ref{item:u2}). Understanding whether the linear dependence of \Cref{corollary:near-opt-swap} on $n$ is necessary in general games is left as an interesting open question.

\begin{theorem}[Refinement for Large Games]
    \label{theorem:largegames}
    Suppose that all players use $\bmoftrl$. Furthermore, assume that the utility of player $i \in \range{n}$ depends on a subset of players $\cN_i \subseteq \range{n}$, so that 
    \begin{enumerate}
        \item $|\cN_i| \leq c \leq n -1$; and \label{item:u1}
        \item $\max_{i \in \range{n}} | \{ j \neq i : i \in \cN_j \}| \leq c$. \label{item:u2}
    \end{enumerate}
    Then, for $\eta = \frac{1}{128 c \max_{j \in \range{n}} \{ \sqrt{m_j} \} }$,
    \begin{equation*}
        \sum_{i=1}^n \swapreg_i^T \leq 256 c \max_{j \in \range{n}} \{ \sqrt{m_j} \} \sum_{j=1}^n m_j^2 \log T.
    \end{equation*}
    Moreover, for $\eta = \frac{1}{128 \sqrt{c n} \max_{j \in \range{n}} \{ \sqrt{m_j} \}} \leq \frac{1}{128 c \max_{j \in \range{n}} \{ \sqrt{m_j}\}}$,
    \begin{equation*}
    \swapreg_i^T \leq 256 \max_{j \in \range{n}} \{ \sqrt{m_j} \} \left( \sqrt{c n}  m_i^2 + \sqrt{\frac{c}{n}} \sum_{j=1}^n m_j^2 \right) \log T.
\end{equation*}
In particular, if $m_i = m$ for all $i \in \range{n}$,
\begin{equation*}
    \swapreg_i^T \leq 512 \sqrt{c n} m^{5/2} \log T.
\end{equation*}
\end{theorem}

\begin{proof}
The proof proceeds similarly to the proof of \Cref{theorem:log-bounded_trajectories}. First, we have that
\begin{equation*}
    \left( \| \vec{u}_i^{(t)} - \vec{u}_i^{(t-1)} \|_\infty \right)^2 \leq \left( \sum_{j \in \cN_i} \| \vec{x}_{j}^{(t)} - \vec{x}_{j}^{(t-1)} \|_1 \right)^2 \leq c \sum_{j \neq i} \| \vec{x}_{j}^{(t)} - \vec{x}_{j}^{(t-1)} \|^2_1,
\end{equation*}
since $|\cN_i| \leq c$. Thus, using \Cref{theorem:rvu-swap}, $\sum_{i=1}^n \swapreg_i^T$ can be upper bounded by
\begin{align}
&2 \frac{\log T }{\eta} \sum_{i=1}^n m_i^2 + 4\eta c \sum_{i =1}^n \sum_{j \in \cN_i} \sum_{t=1}^T \| \vec{x}_{j}^{(t)} - \vec{x}_{j}^{(t-1)} \|^2_1 - \sum_{i = 1}^n \frac{1}{2048 m_i \eta} \sum_{t=1}^T  \| \Vec{x}_i^{(t)} - \Vec{x}_i^{(t-1)}\|_1^2 \notag \\
\leq &2 \frac{\log T }{\eta} \sum_{i=1}^n m_i^2 + \sum_{i=1}^n \left( 4\eta c^2 - \frac{1}{2048 m_i \eta} \right) \sum_{t=1}^T \| \vec{x}_i^{(t)} - \vec{x}_i^{(t-1)} \|_1^2 \label{align:graph} \\
\leq &2 \frac{\log T }{\eta} \sum_{i=1}^n m_i^2 - \frac{1}{4096 \eta} \sum_{i=1}^n \frac{1}{m_i} \sum_{t=1}^T \| \vec{x}_i^{(t)} - \vec{x}_i^{(t-1)} \|_1^2, \label{align:eta}
\end{align}
where \eqref{align:graph} uses the assumption that $| \{ j \neq i : i \in \cN_j \}| \leq c$, for any player $i \in \range{n}$, and \eqref{align:eta} follows since $\eta \leq \frac{1}{128 c \sqrt{m_i}}$ for all $i \in \range{n}$. As a result, for $\eta = \frac{1}{128 c \max \{\sqrt{m_j}\}}$,
\begin{equation*}
    \sum_{i=1}^n \swapreg_i^T \leq 256 c \max_{j \in \range{n}} \{ \sqrt{m_j} \} \sum_{j=1}^n m_j^2 \log T.
\end{equation*}
Furthermore, given that $\sum_{i=1}^n \swapreg_i^T \geq 0$,
\begin{equation*}
    \frac{1}{\max_{j \in \range{n}} \{ m_j \}} \sum_{i=1}^n \sum_{t=1}^T \| \vec{x}_i^{(t)} - \vec{x}_i^{(t-1)}\|_1^2 \le \sum_{i=1}^n \frac{1}{m_i} \sum_{t=1}^T \| \vec{x}_i^{(t)} - \vec{x}_i^{(t-1)}\|_1^2 \leq 8192 \sum_{i=1}^n m_i^2 \log T.
\end{equation*}
Thus, for $\eta = \frac{1}{128 \sqrt{c n} \max_{j \in \range{n}} \{ \sqrt{m_j} \}} \leq \frac{1}{128 c \max_{j \in \range{n}} \{ \sqrt{m_j}\}}$,
\begin{equation*}
    \swapreg_i^T \leq 256 \sqrt{c n} \max_{j \in \range{n}} \{ \sqrt{m_j} \} m_i^2 \log T + 256 \sqrt{\frac{c}{n}} \max_{j \in \range{n}} \{ \sqrt{m_j} \} \sum_{j=1}^n m_j^2 \log T.
\end{equation*}
\end{proof}
Hence, when $c$ is a small constant this theorem implies an improvement of $\Theta(n)$ for the sum of the players' swap regrets, as well as an $\Theta(\sqrt{n})$ factor for each individual swap regret.

%% file: text/appendix_experiments.tex
\section{Experiments}
\label{appendix:experiments}

In this section we include additional experiments in order to corroborate some of our theoretical results. First, regarding \Cref{fig:Shapley} in the main body, we considered a bimatrix (general-sum) game described with the following payoff matrices.

\begin{equation}
    \label{eq:Shapley}
    \mat{A} = 
    \begin{pmatrix}
    0 & 0.5 & 1.5 \\
    1.5 & 0 & 1 \\
    0.5 & 1.5 & 0
    \end{pmatrix};
    \quad\quad
    \mat{B} = 
    \begin{pmatrix}
    0 & 1.5 & 1 \\
    1 & 0 & 1.5 \\
    1.5 & 1 & 0
    \end{pmatrix}.
\end{equation}

This game is a slight variant of \emph{Shapley's game}~\citep{Shapley64:Some}, a general-sum two-player game used by Shapley in order to illustrate that fictitious play does not converge to Nash equilibria in general-sum games. Shapley's game is not suited to illustrate the cycling behavior of the dynamics in our case since it has a unique Nash equilibrium, occurring when both players play uniformly at random; as such, \eqref{eq:OFTRL} is initialized at the equilibrium. On the other hand, the (unique) equilibrium of the game described in \eqref{eq:Shapley} occurs when $\Vec{x}^* = (\frac{1}{3}, \frac{1}{3}, \frac{1}{3})$ and $\Vec{y}^* = (\frac{1}{4}, \frac{2}{5}, \frac{7}{20})$~\citep{Avis10:Enumeration}. As illustrated in~\Cref{fig:Shapley}, $\bmoftrl$ does not appear to converge to a Nash equilibrium---at least in a  \emph{last-iterate} sense. In contrast, we conjecture that the last iterate of \eqref{eq:OFTRL} with log-barrier regularization converges to the set of Nash equilibria in zero-sum games, and this property seems plausible even under the BM construction. 

Moreover, we conduct experiments on random $3 \times 3$ bimatrix (normal-form) general-sum games. Specifically, each entry of the payoff matrices is an independent random variable drawn from the uniform distribution in $[-1, 1]$. In \Cref{fig:swap_regrets} we illustrate the swap regret of the $\bmoftrl$ algorithm with a time-invariant learning rate $\eta = 0.1$.

\begin{figure}[!ht]
    \centering
    \includegraphics[scale=0.63]{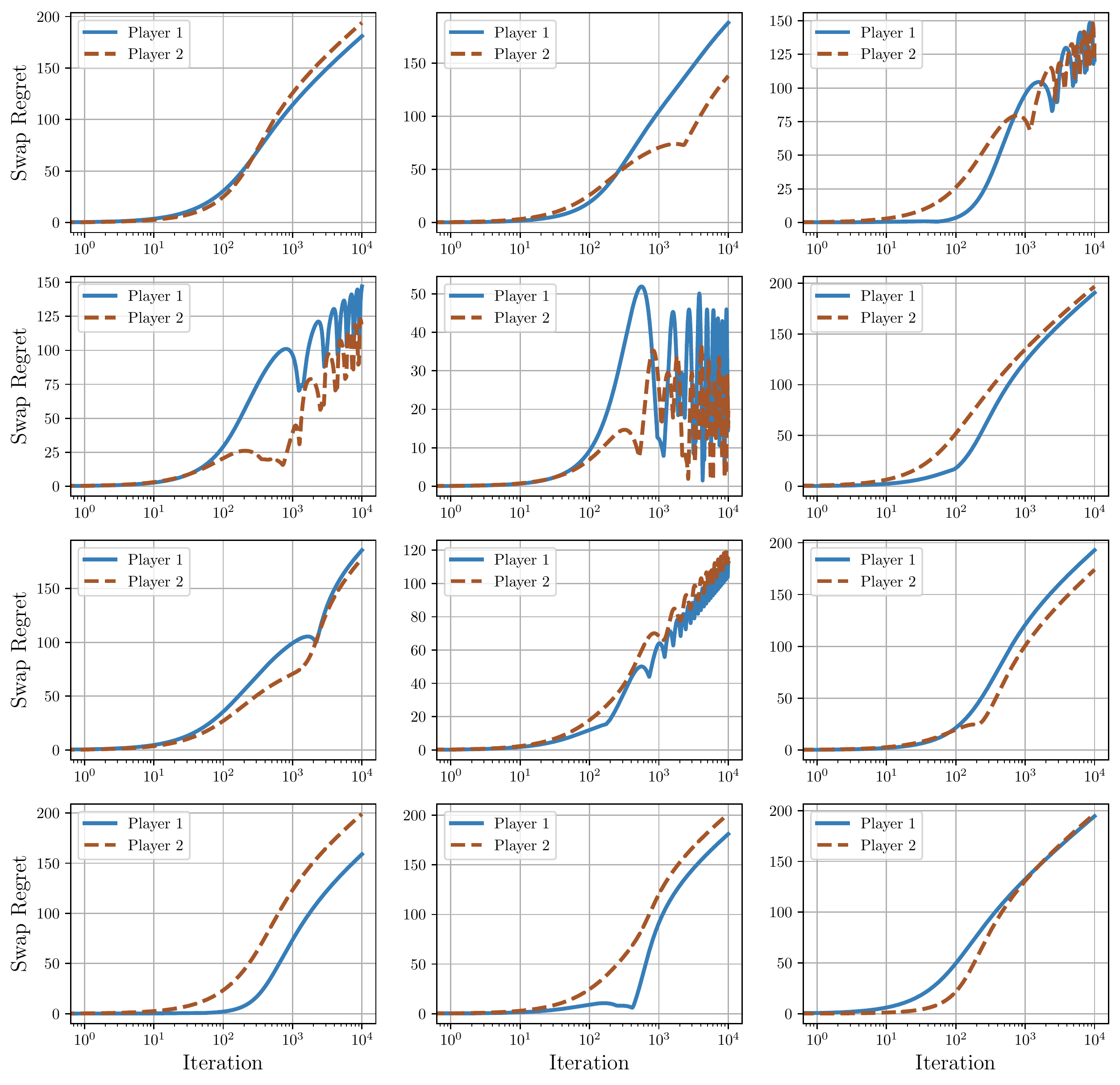}
    \caption{The swap regret experienced by each player for $T = 10^4$ iterations when both players employ $\bmoftrl$ with $\eta = 0.1$. Each plot corresponds to a random $3 \times 3$ bimatrix game. The $x$-axis represents the iteration, in logarithmic scale, while the $y$-axis shows the swap regret experienced by each player at the given iteration. These results corroborate the $O(\log T)$ rates established in \Cref{corollary:near-opt-swap}, showing that our analysis is essentially tight.}
    \label{fig:swap_regrets}
\end{figure}

